\pdfoutput=1
\documentclass[11pt, a4paper]{article}
\usepackage[utf8]{inputenc}
\usepackage[T1]{fontenc}
\usepackage[english]{babel}

\usepackage{amsmath,amsthm,amssymb,amsfonts,mathtools}
\allowdisplaybreaks
\usepackage{graphicx}

\usepackage{appendix}
\usepackage{placeins}
\usepackage{cancel}

\usepackage[a4paper, margin=2cm]{geometry}
\usepackage[table,xcdraw]{xcolor}
\usepackage{hhline}

\usepackage{caption}
\usepackage{subcaption}
\usepackage{csquotes}
\usepackage{natbib}

\usepackage[ruled, linesnumbered, noalgohanging]{algorithm2e}

\makeatletter
\g@addto@macro\normalsize{%
	\setlength\abovedisplayskip{13pt}
	\setlength\belowdisplayskip{13pt}
	\setlength\abovedisplayshortskip{13pt}
	\setlength\belowdisplayshortskip{13pt}
}
\makeatother

\newtheorem{theorem}{Theorem}[section]

\usepackage{multirow}
\usepackage{booktabs}

\usepackage[pdfencoding=auto,colorlinks,citecolor=blue]{hyperref}
\usepackage{cleveref}

\title{Bayesian inference, model selection and  likelihood estimation using fast rejection sampling: the Conway-Maxwell-Poisson distribution}
\author{Alan Benson\textsuperscript{1,2,$\star$} and Nial Friel\textsuperscript{1,2}}
\makeatletter
\renewcommand\@date{{%
		\vspace{-\baselineskip}%
		\large\centering
		\bigskip
		\textsuperscript{1}School of Mathematics and Statistics, University College Dublin\par
		\vspace{0.5ex}
		\textsuperscript{2}Insight Centre for Data Analytics\par
		\vspace{0.5ex}
		\textsuperscript{$\star$}alan.benson@insight-centre.org
		
		\bigskip
		\bigskip
		
		\today
}}
\makeatother

\usepackage{bm}

\begin{document}
\maketitle

\begin{abstract}
Bayesian inference for models with intractable likelihood functions represents a challenging suite of problems in modern statistics. In this work we analyse the 
Conway-Maxwell-Poisson (COM-Poisson) distribution, a two parameter generalisation of the Poisson distribution. COM-Poisson regression modelling allows the 
flexibility to model dispersed count data as part of a generalised linear model (GLM) with a COM-Poisson response, where exogenous covariates control the mean and 
dispersion level of the response. The major difficulty with COM-Poisson regression is that the likelihood function contains multiple intractable normalising 
constants and is not amenable to standard inference and MCMC techniques. Recent work by \cite{chanialidis2017} has seen the development of a sampler to draw 
random variates from the COM-Poisson likelihood using a rejection sampling algorithm. We provide a new rejection sampler for the COM-Poisson distribution which significantly reduces the CPU time required to perform inference for COM-Poisson regression models. An extension of this work shows that for any intractable likelihood function with an associated rejection sampler it is possible to construct unbiased estimators of the intractable likelihood which proves useful for model selection or for use within pseudo-marginal MCMC algorithms \citep{andrieu2009}. We demonstrate all of these methods on a real-world dataset of takeover bids.
\end{abstract}

\section{Introduction}
Modelling count data continues to be a important area in the practice of statistics. 
Often covariate information is available which may prove useful in explaining the counts, for example, time of day influencing the number calls to a particular call centre or amount borrowed influencing the number of loan defaults by a consumer. Poisson regression is a standard framework for modelling covariate-dependent count data allowing the mean response to depend on the covariates through a log-linear link function. In practice, Poisson regression fails to capture the phenomenon of dispersion, where the mean and variance of the response differ significantly. Dispersion measures the spread of a random variable and is quantified through the ratio of the variance of a random variable to its mean. Count data often exhibits underdispersion (variance less than the mean), overdispersion (variance greater than the mean) or equidispersion (equality of mean and variance) but Poisson regression is only capable of modelling equidispersed data. Many alternative approaches have been suggested to capture dispersion, such as using quasilikelihood or the popular negative binomial regression \citep{hilbe2011} which allows modelling of overdispersed data only. The COM-Poisson distribution \citep{conway1962} was first introduced to model queuing systems and was then revived as a statistical model for dispersed count data by \cite{shmueli2005}. This revival lead to the introduction of COM-Poisson regression models by \cite{guikema2008} as a flexible regression model for dispersed count data with covariates. However, this flexibility comes with the caveat that the COM-Poisson likelihood function contains an intractable normalising term which requires the need for non-standard methods to conduct parameter inference. This intractability has restricted the popularity of COM-Poisson regression as a viable solution to modelling dispersion. Recent work by \cite{chanialidis2017} has seen a Bayesian analysis of COM-Poisson regression under the framework of doubly-intractable Bayesian inference, where the authors overcome the intractability through the use of a rejection sampling algorithm for the COM-Poisson distribution. This rejection sampler is then used within the exchange algorithm of \citet{murray2006} to perform Bayesian parameter inference.

Our novel contribution in this paper is to provide a more efficient approach to that of \cite{chanialidis2017} by designing a simpler rejection sampler for the 
COM-Poisson distribution which is computationally less intensive to sample from and leads to faster inference within COM-Poisson regression models. In one example 
in this paper, we observe a greater than threefold reduction in CPU time by using our rejection sampler. The rejection sampler of \citet{chanialidis2017} is 
derived from a discrete version of adaptive rejection sampling \citep{gilks1992} and requires the costly construction of a piecewise truncated geometric enveloping 
distribution to tightly encapsulate the COM-Poisson probability density function. Our rejection sampler requires only a single envelope distribution depending on 
the dispersion parameter of the COM-Poisson distribution. This single envelope distribution does come with a higher rejection rate but completely bypasses the 
costly setup and sampling cost involved in piecewise truncated geometric envelope distributions. A further and important novelty of our work is to show that when a rejection 
sampling algorithm is available to sample from an intractable likelihood function of interest, then it is possible to construct an unbiased estimate of this 
intractable likelihood function using the rejection sampler. Our idea works by exploiting a link between rejection sampling efficiency and its relationship 
with the reciprocal intractable normalising constant of the likelihood function.

The remainder of this paper is organised as follows: Section \ref{sec:glms} reviews the COM-Poisson distribution and COM-Poisson regression models under the framework of doubly-intractable Bayesian inference.
Section \ref{sec:rejsampling} reviews rejection sampling for intractable likelihoods and our new COM-Poisson rejection sampler is introduced in Section \ref{sec:cmpsampler}. Section \ref{sec:unbiasedlike} describes the novel method of constructing unbiased estimators of intractable likelihood functions through rejection sampling. Section \ref{sec:pseudomarginal} shows how pseudo-marginal MCMC can be performed using this unbiased intractable likelihood estimator. Section \ref{sec:results} shows results of these methods applied firstly to a dataset with no covariates in Section \ref{sec:inventory} and then to a real-world COM-Poisson regression dataset of company takeover bids in Section \ref{sec:takeover}. Finally, Section \ref{sec:discussion} gives discussion and some thoughts for future work.

\section{Bayesian doubly-intractable inference for GLMs}
\label{sec:glms}
The COM-Poisson distribution \citep{conway1962} is a two-parameter discrete probability distribution that allows the flexibility to model count data with 
over-, under- and equi-dispersion. The probability mass function for a COM-Poisson random variable $Y$ with parameters $\mu > 0$ and $\nu \geq 0$ is defined over the 
non-negative integers as
\begin{equation*}
f(y \vert \mu, \nu) = \left( \frac{\mu^y}{y!} \right)^\nu \frac{1}{\mathcal{Z}_f(\mu, \nu)}.
\end{equation*}
The unnormalised component of the mass function is denoted $q_f(y \vert \mu, \nu) = \left( \frac{\mu^y}{y!} \right)^\nu$ and 
$\mathcal{Z}_f(\mu, \nu) = \sum_{y=0}^{\infty} q_f(y \vert \mu, \nu)$ is an intractable normalising constant, having no closed form representation for this model. 
The mode of the COM-Poisson distribution is $\lfloor \mu \rfloor$, with two consecutive modes at $\mu$ and $\mu-1$ in the case where $\mu$ is an integer. The moments 
of the COM-Poisson distribution are unavailable directly due to the intractable normalising constant. The moments can however be approximated through the use of an 
asymptotic representation of $\mathcal{Z}_f(\mu, \nu)$ derived by \citet{shmueli2005}. The approximations for the mean and variance using this asymptotic representation 
are, respectively, 
\begin{equation}
\mathrm{E}(Y) \approx \mu + \frac{1}{2\nu} - \frac{1}{2}, \quad \mathrm{Var}(Y) \approx \frac{\mu}{\nu}.
\label{eq:meanvarapprox}
\end{equation}
These approximations are quite accurate for a wide range of $\mu$ and $\nu$, except for small $\mu$ or small $\nu$ for the case of $\mathrm{E}(Y)$ and for $\nu>1$, regardless of the value of $\mu$ for the case of $\mathrm{Var}(Y)$.
The purpose of the parameter $\nu$ is to control dispersion through this parameter's inverse relationship with the variance as seen in \eqref{eq:meanvarapprox}. When $\nu = 1$ the distribution exhibits equidispersion and the probability mass function reduces to that of a Poisson random variable with $\mathcal{Z}_f(\mu, 1) = \exp(\mu)$. Alternatively setting $\nu < 1$ gives overdispersion and setting $\nu>1$ gives underdispersion.

The main application of the COM-Poisson distribution is to GLMs, referred to as COM-Poisson regression models, where $n$ responses $y_{1:n} = (y_1, \dots y_n)$ are 
assumed to follow a COM-Poisson distribution with the parameters $\mu$ and $\nu$ of each response being conditional on $b$ available covariates 
$x_i = \left(x_{i1}, \dots, x_{ib} \right)$ for $i=1,\dots n$. Due to this conditioning we now index the parameters for observation $y_i$ as $\mu_i$ and $\nu_i$. 
COM-Poisson regression modelling was first introduced by \citet{guikema2008} and \cite{sellers2010} in an effort to extend Poisson regression by allowing covariates 
to influence not only the location parameter $\mu_i$ but also the dispersion parameter $\nu_i$. For notational simplicity we now let $\theta_i = (\mu_i, \nu_i)^T$ be 
the parameter for observation $y_i$ with $\theta_i$ derived from setting $\theta_i = \eta(\beta, x_i)$ where $\eta(\beta, x_i)$ is a link function of the $b$ covariates 
and $\beta$ are the parameters of this link function. The major difficulty is that the likelihood of the COM-Poisson regression model involves multiple intractable 
normalising constants, 
\begin{equation}
f(y_{1:n} \vert \theta_{1:n}) = \prod_{i=1}^{n}f(y_i \vert \theta_{i}) =  \prod_{i=1}^{n}\frac{q_f(y_i \vert \theta_{i})}{\mathcal{Z}_f(\theta_{i})}, \quad \theta_i = \eta(\beta,x_i).
\label{eq:fulllike}
\end{equation}
The particular link function $\eta(\beta, x_i)$ for the COM-Poisson was specified by \citet{guikema2008} as a dual-link function to allow the available covariates to enter both the $\mu_i$ and $\nu_i$ parameter. In this dual-link function scenario the parameter $\beta$ is split into two sub-vectors as $\beta = (\beta_\mu, \beta_\nu)$, where $\beta_\mu$ are the coefficients for the $\mu$ link component and $\beta_\nu$ are the coefficients for the $\nu$ link component. The dual-link function is then
\begin{equation*}
\theta_i  = \begin{bmatrix}
\mu_i \\
\nu_i
\end{bmatrix}  =\eta(\beta, x_i) = 
\begin{bmatrix}
 \eta_1(\beta_\mu, x_i) \\
 \eta_2(\beta_\nu, x_i)
\end{bmatrix} 
  = \begin{bmatrix}
\exp \left(\beta_{\mu,0} + \sum_{j=1}^{b} \beta_{\mu,j}x_{ij} \right) \\
\exp \left(\beta_{\nu,0} + \sum_{j=1}^{b} \beta_{\nu,j} x_{ij} \right)
\end{bmatrix},
\end{equation*}
where some $\beta_{\mu}$ are set fixed to 0 when the covariate is not included in $\eta_1(\beta_\mu, x_i)$, similarly for $\eta_2(\beta_\nu, x_i)$.

Bayesian parameter inference for $\beta$ in COM-Poisson regression models is a doubly-intractable problem as we now demonstrate. Placing a prior on $\beta$, the posterior distribution for $\beta$ is formed with this prior and the complete likelihood \eqref{eq:fulllike} as
\begin{align}
\pi(\beta \vert y_{1:n}) &\propto f(y_{1:n} \vert \theta_{1:n}) \pi(\beta) \nonumber \\
&= \prod_{i=1}^{n}\frac{q_f(y_i \vert \theta_{i})}{\mathcal{Z}_f(\theta_{i})} \pi(\beta), \quad \theta_i = \eta(\beta,x_i). 
\label{eq:posterior}
\end{align}
This posterior is termed doubly-intractable since the likelihood is intractable, containing $n$ intractable normalising constants, and the posterior itself cannot 
be normalised. Without the ability to evaluate $f(y_{1:n} \vert \theta_{1:n})$ pointwise, the direct application of standard MCMC methods is infeasible. For example, a 
Metropolis-Hastings algorithm proposing a move in the link function from $\beta$ to $\beta^\prime$ using a proposal distribution $h(\beta, \beta^\prime)$ requires evaluation of the intractable ratios $\left\{\frac{\mathcal{Z}_f(\theta_{i})}{\mathcal{Z}_f(\theta^\prime_{i})} \right\}_{i=1}^{n}$ to compute the acceptance ratio
\begin{equation}
\alpha(\beta, \beta^\prime) = \min \left\{ 1, \frac{\prod_{i=1}^{n}\frac{q_f(y_i \vert \theta^\prime_{i})}{\mathcal{Z}_f(\theta^\prime_{i})}}{\prod_{i=1}^{n}\frac{q_f(y_i \vert \theta^{}_{i})}{\mathcal{Z}_f(\theta^{}_{i})}}
\frac{h(\beta^\prime, \beta)}{h(\beta, \beta^\prime)}
\frac{\pi(\beta^\prime)}{\pi(\beta)} \right\}.
\label{eq:acceptratio}
\end{equation}
Possible methods to compute this acceptance ratio include replacing the intractable ratios with approximations using the asymptotic approximation from \cite{shmueli2007} or by a truncated sum $\mathcal{Z}_f(\theta_i) = \sum_{i=0}^{k} q_f(y \vert \theta_i)$ but both of these methods introduce some bias into the acceptance ratio.

An ingenious solution to the issue of computing the acceptance ratio \eqref{eq:acceptratio} was proposed by \citet{murray2006} drawing on earlier work by \cite{moeller2006}. The solution, known as the exchange algorithm, augments the doubly-intractable posterior \eqref{eq:posterior} with auxiliary data, provided that it is possible to sample data exactly from the intractable likelihood. The original exchange algorithm overcomes a single intractable ratio i.e. $n=1$ however the algorithm can easily be extended to the case of $n > 1$ intractable ratios as we now show. Starting from the posterior \eqref{eq:posterior} which contains the product of $n$ independent non-identical intractable likelihood functions, the exchange algorithm proposes, as in standard Metropolis-Hastings, to update the parameter from the current state $\beta$ to proposed state $\beta^\prime$ using $h(\beta, \beta^\prime)$ but in addition the posterior is augmented with $n$ auxiliary draws $y_{1:n}^\prime = (y^\prime_1, \dots y^\prime_n)$ generated from the likelihood at the proposed parameters $\theta^\prime_{i} = \eta(\beta^\prime, x_{i})$. The augmented posterior is written
\begin{equation}
\pi(\beta, \beta^\prime, y_{1:n}^\prime \vert y_{1:n}) \propto \underbrace{f(y_{1:n} \vert \theta_{1:n})}_{\text{likelihood}}  \pi(\beta) h(\beta, \beta^\prime)
\underbrace{f(y^\prime_{1:n} \vert \theta^{\prime}_{1:n})}_\text{auxiliary draws} 
\label{eq:augposterior}
\end{equation}
and the marginal of this augmented posterior for $\beta$ is the target posterior of interest \eqref{eq:posterior}. The algorithm can be seen as a Markov chain over an augmented parameter $(\beta, y_{1:n})$.
The acceptance ratio for the augmented posterior \eqref{eq:augposterior} is now tractable due to the cancellation of all intractable normalising constants
\begin{align}
\alpha_{\text{exchange},n}(\beta, \beta^\prime) = \min \left\{1, \frac{\prod_{i=1}^n\frac{q_f(y_i \vert \theta^\prime_{i})}{\mathcal{Z}_f(\theta^\prime_{i})}}{\prod_{i=1}^n\frac{q_f(y_i \vert \theta^{}_{i})}{\mathcal{Z}_f(\theta^{}_{i})}}
\frac{h(\beta^\prime, \beta)}{h(\beta, \beta^\prime)}
\frac{\pi(\beta^\prime)}{\pi(\beta)}
\frac{\prod_{i=1}^n\frac{q_f(y^\prime_i \vert \theta^{}_{i})}{\mathcal{Z}_f(\theta^{}_{i})}}{\prod_{i=1}^n\frac{q_f(y^\prime_i \vert \theta^\prime_{i})}{\mathcal{Z}_f(\theta^\prime_{i})}}\right\}, \quad y^\prime_i \sim f(\cdot \vert \theta^\prime_{i}), \, \forall i \nonumber \\[2ex]
=
\min \left \{1, \frac{\prod_{i=1}^nq_f(y_i \vert \theta^\prime_{i})}{\prod_{i=1}^n q_f(y_i \vert \theta_{i})}
\frac{h(\beta^\prime, \beta)}{h(\beta, \beta^\prime)}
\frac{\pi(\beta^\prime)}{\pi(\beta)}
\frac{\prod_{i=1}^n q_f(y^\prime_i \vert \theta_{i})}{\prod_{i=1}^n q_f(y^\prime_i \vert \theta^\prime_{i})}
\xcancel{
\frac{\prod_{1=1}^{n}\frac{1}{\mathcal{Z}_f \left(\theta^\prime_{i} \right)}}{\prod_{1=1}^{n}\frac{1}{\mathcal{Z}_f \left(\theta_{i} \right)}}
\frac{\prod_{1=1}^{n}\frac{1}{\mathcal{Z}_f \left(\theta_{i} \right)}}{\prod_{1=1}^{n}\frac{1}{\mathcal{Z}_f \left(\theta^\prime_{i} \right)}}
}
\right \}
\label{eq:acceptexchange}
\end{align}
The clever cancellation of the normalising constants in \eqref{eq:acceptexchange} is due to the \textit{exchange} of parameters ($\theta_{i}, \theta^\prime_{i}$) associated with the observed data ($y_{1:n}$) and the auxiliary data ($y_{1:n}^\prime$), the auxiliary being data discarded after each move. The major requirement in the exchange algorithm is the ability to generate the auxiliary data, requiring exact draws from the likelihood at each proposed parameter. Perfect sampling \citep{propp1996} from the likelihood is one possible method to do this but can be prohibitively expensive for many models. In certain models sampling from an intractable likelihood may even be as difficult as computing the intractable normalising constant itself and instead one could resort to sampling using an approximate likelihood sampler such as in the case of the exponential random graph model \citep{caimo2011}. The COM-Poisson distribution was shown by \citet{chanialidis2017} to be amenable to rejection sampling, allowing the auxiliary data to be generated. In Section \ref{sec:cmpsampler} we present our faster rejection sampler having a cheaper sampling cost than the sampler of \citet{chanialidis2017}. Rejection sampling is not applicable to all intractable likelihoods but when rejection sampling is available we show in Section \ref{sec:unbiasedlike} that the intractable likelihood can be estimated without bias through an unbiased estimate of the reciprocal normalising constant. This unbiased intractable likelihood estimator is also guaranteed to be positive which differs from other estimates using Russian roulette stochastic truncation \citep{lyne2015} or truncated Markov chains \citep{wei2016}.

\section{Rejection sampling to an intractable likelihood estimator}
\label{sec:rejsampling}
Rejection sampling \citep{vonneumann1951} is a sampling scheme to generate statistically independent samples from a target probability distribution of interest. We denote the target density (likelihood) as
\begin{equation*}
f(y \vert \theta) = \frac{q_f(y \vert \theta)}{\mathcal{Z}_f(\theta)}
\end{equation*}
where in this section we drop the observation index $i$ for ease of notation.
 This target density may be tractable or intractable but is assumed in any case to be difficult to sample from. The rejection sampling method works by constructing 
 an envelope distribution over the target density, proposing samples from this envelope distribution and choosing to accept a portion of these samples. The envelope 
 distribution should be chosen so that it is both computationally efficient to  sample from and matches closely in shape the target density. The portion of samples 
 that will be accepted depends on the similarity of the envelope and the target distribution. The closer the envelope distribution matches the target distribution, 
 the higher the number of samples that will be accepted. To generate a single draw $y$ from $f(y \vert \theta)$ using rejection sampling, we consider an envelope distribution $g(y \vert \gamma)$ where $\gamma \in \Gamma$ is its parameter. We assume that the envelope density can be written in a similar form to $f(y \vert \theta)$ as
\begin{equation*}
g(y \vert \gamma) = \frac{q_g(y \vert \gamma)}{\mathcal{Z}_g(\gamma)}, \quad \text{with} \hspace{0.5em} \mathcal{Z}_g(\gamma) = \int_{y} q_g(y \vert \gamma) \, dy.
\end{equation*}
The envelope distribution's normalising constant $\mathcal{Z}_g(\gamma)$ may be either tractable or intractable. A necessary condition for $g (\cdot\vert \gamma)$ is that it dominates the support of $f(\cdot\vert\theta)$  i.e.\ $g(y \vert \gamma) = 0 \Longrightarrow f(y \vert \theta) = 0$. It is also necessary to bound $f(y\vert\theta)$ using $g(y \vert \gamma)$ and a positive finite bounding constant $M$ that satisfies the envelope inequality $Mg(y \vert \gamma) > f(y \vert \theta), \, \forall y$. Finding such a constant can be a difficult task even for tractable densities. The optimal value of $M$ is denoted as $M_{f/g} = \sup_y \left\{ \frac{f(y \vert \theta)}{g(y \vert \gamma)} \right\}$. In the case of intractable likelihoods, this optimal constant is impossible to evaluate since at least one of either $\mathcal{Z}_f(\theta)$ or $\mathcal{Z}_g(\gamma)$ is unknown leaving $M_{f/g}$ intractable. In the case where both $f(y \vert \theta)$ and $g(y \vert \gamma)$ are tractable distributions and $M_{f/g}$ can be found, the vanilla rejection sampling algorithm to obtain a single draw from $f(y \vert \theta)$ proceeds as in Algorithm \ref{alg:vanillarejectionsampling}.
\begin{algorithm}[htbp]
	\linespread{1.2}\selectfont
	\SetKwInput{Input}{Input}
	\Input{Target distribution $f(\cdot \vert \theta)$, envelope distribution $g(\cdot \vert \gamma)$ and $M_{f/g} = \sup_y \left\{ \frac{f(y \vert \theta)}{g(y \vert \gamma)}\right\}$}
	Sample a proposal $y^{\ast}$ from $g(\cdot \vert \gamma)$\\
	Accept $y^{\ast}$ as a draw from $f(\cdot \vert \theta)$ by a Bernoulli trial with acceptance probability 
	\begin{equation}
	\alpha(y^{\ast}) = \frac{f(y^{\ast} \vert \theta)}{M_{f/g}g(y^{\ast} \vert \gamma)},
	\label{eq:acceptystarprob}
	\end{equation} otherwise repeat. \\
	\caption{Rejection sampling}
	\label{alg:vanillarejectionsampling}
\end{algorithm}

Algorithm \ref{alg:vanillarejectionsampling} introduces the rejection sampling Bernoulli trial with the outcome conditional on the proposed $y^{\ast}$ through the acceptance probability $\alpha(y^\ast)$. This Bernoulli trial is at the core of the rejection sampling method as it decides for a given $y^{\ast}$ what proportion of samples from $g(\cdot \vert \gamma)$ to accept. Once intractability is introduced into this Bernoulli trial through $M_{f/g}$, it would appear that $\alpha(y^{\ast})$ becomes intractable but this is fortunately not the case. The intractable bounding constant $M_{f/g}$ can be decomposed into tractable and intractable components by extracting the reciprocal normalising constants from $M_{f/g}$ as follows
\begin{equation}
M_{f/g} = \sup_y \left\{ \frac{f(y \vert \theta)}{g(y \vert \gamma)} \right\} = \frac{1/\mathcal{Z}_f(\theta)}{1/\mathcal{Z}_g(\gamma)} \sup_y \left\{ \frac{q_f(y \vert \theta)}{q_g(y \vert \gamma)} \right\} = \frac{\mathcal{Z}_g(\gamma)}{\mathcal{Z}_f(\theta)}B_{f/g}
\label{eq:mfgbfg}
\end{equation}
and this introduces the tractable bounding constant $B_{f/g} =  \sup_y \left\{ \frac{q_f(y \vert \theta)}{q_g(y \vert \gamma)} \right\}$. $B_{f/g}$ can be viewed as the upper bound on the ratio of the unnormalised densities $q_f(\cdot\vert\theta)$ and $q_g(\cdot\vert\gamma)$, whereas $M_{f/g}$ is the upper bound on the associated normalised densities $f(\cdot\vert\theta)$ and $g(\cdot\vert\gamma)$. 
The advantage of rejection sampling algorithms is that they can be run without full knowledge of $M_{f/g}$, it is sufficient to know $B_{f/g}$ which is usually more readily available. Substituting the expression $M_{f/g} =\frac{\mathcal{Z}_g(\gamma)}{\mathcal{Z}_f(\theta)}B_{f/g}$ into the acceptance probability \eqref{eq:acceptystarprob} gives the tractable acceptance probability for the conditional Bernoulli trial as
\begin{equation*}
\alpha(y^{\ast}) = \frac{f(y^{\ast} \vert \theta)}{\left({\frac{\mathcal{Z}_g(\gamma)}{\mathcal{Z}_f(\theta)}B_{f/g}}\right)g(y^{\ast} \vert \gamma)} = \frac{q_f(y^{\ast} \vert \theta)\frac{1}{\mathcal{Z}_f(\theta)}}{\left({\frac{\mathcal{Z}_g(\gamma)}{\mathcal{Z}_f(\theta)}B_{f/g}}\right)q_g(y^{\ast} \vert \gamma)\frac{1}{\mathcal{Z}_g(\gamma)}} = \frac{q_f(y^{\ast}\vert \theta)}{B_{f/g}q_g(y^{\ast}\vert \gamma)}.
\end{equation*}
It is important to note that $\mathcal{Z}_g(\gamma)$ is also not required to be tractable, although it typically is for simple envelope distributions. A case where $\mathcal{Z}_g(\gamma)$ is intractable would be when the envelope distribution itself requires sampling by another rejection sampler and this opens the possibility of chaining many rejection samplers together to sample from evermore complex target distributions. The main purpose of rejection sampling in this paper is to provide  a means to generate the auxiliary draws required within the exchange algorithm acceptance ratio \eqref{eq:acceptexchange} allowing the doubly-intractable inference to be performed but Section \ref{sec:unbiasedlike} will show another novel usage for rejection sampling. Before discussing this, we demonstrate our COM-Poisson rejection sampler in the next section.

\subsection{COM-Poisson rejection sampler}
\label{sec:cmpsampler}
For the COM-Poisson distribution, our sampler uses a choice  of two envelope distributions dependent on the value of $\nu$. This choice is motivated by Figure~\ref{fig:envelopes}, where a Poisson distribution is used in the case of $\nu \geq 1$ and a geometric distribution in the case of $\nu < 1$. Theorem \ref{thm:exactcmpsampler} demonstrates and proves why this works. The COM-Poisson rejection sampling algorithm is given in pseudocode after this theorem.

\begin{theorem}[COM-Poisson intractable rejection sampler]
	\label{thm:exactcmpsampler}
	Suppose that $Y \sim g_1(y \vert \gamma)$, a Poisson random variable with parameter $\gamma = \mu$ and $Y \sim g_2(y \vert \gamma)$, a geometric random variable with parameter $\gamma = p$, 
	for some $0<p<1$, are used as two enveloping distributions as follows,
	\begin{align}
	g(y \vert \gamma) &= 
	\begin{dcases}
	g_1( y \vert \gamma = p) = p(1-p)^{y}, & \text{if $\nu < 1$ (geometric envelope)},\\
	g_2( y \vert \gamma = \mu) = \dfrac{\mu^y}{e^\mu y!}, & \text{if $\nu \geq 1$ (Poisson envelope)},
	\end{dcases}
	\intertext{together with the following tractable enveloping bounds}
	B_{f/g} &= 
	\begin{dcases}
	\frac{1}{p} \dfrac{{\mu}^{\left(\nu\left\lfloor \frac{\mu}{(1-p)^{{1}/{\nu}}} \right\rfloor\right)}}{(1-p)_{}	^{	\left(\left\lfloor \frac{\mu}{(1-p)^{{1}/{\nu}}} \right\rfloor\right)} \left(\left\lfloor \frac{\mu}{(1-p)^{{1}/{\nu}}} \right\rfloor!\right)	^\nu}, & \text{if $\nu < 1$}, \\
	\left(\dfrac{\mu^{\lfloor \mu \rfloor}}{\lfloor \mu \rfloor!}\right)^{\nu-1}, & \text{if $\nu \geq 1$}.
	\end{dcases}
	\label{eq:bfg}
	\end{align}
	then a sample from the COM-Poisson distribution can be drawn at any parameter value ($\mu, \nu$).	
\end{theorem}
\begin{proof}
	See Appendix \ref{sec:cmpproof}.
\end{proof}

\begin{figure}[htbp]
	\includegraphics[width=\textwidth]{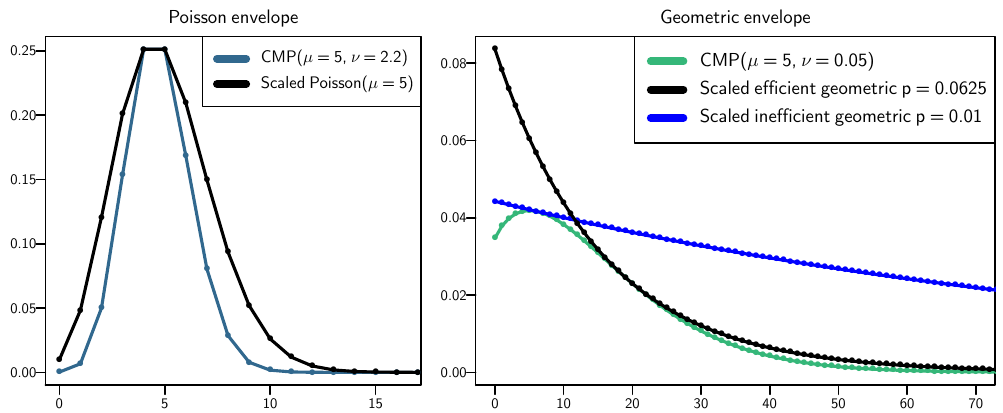}
	\caption{\textbf{Enveloping the COM-Poisson distribution for rejection sampling:} In the underdispersed ($\nu \geq 1$) case (left) a Poisson distribution, having a higher relative dispersion is scaled to envelope the COM-Poisson 
		distribution. In the overdispersed ($\nu < 1$) case (right) a geometric distribution, having a higher relative dispersion is used. The parameter $p$ of the geometric distribution in the overdispersed case controls 
		the efficiency of the enveloping. To choose $p$ we match the first moment of the geometric distribution with the COM-Poisson distribution \eqref{eq:momentmatching}, shown as the black curve (right). The blue curve corresponds 
		to an alternative choice of $p$, showing how an inefficient choice of $p$ results in a geometric envelope with a high rejection region (area between blue and green curves).}
	\label{fig:envelopes}
\end{figure}

Before we present the practical implementation of the COM-Poisson sampling algorithm developed from Theorem~\ref{thm:exactcmpsampler}, it is necessary to provide some brief discussion on the choice of $p$ in the 
geometric envelope. The optimal $p$ parameter, $p^*$ say, would be one for which the acceptance rate of the rejection sampling is maximised. However it is not possible to compute $p^*$ since the acceptance rate, as we will show in Section \ref{sec:unbiasedlike}, involves intractable normalising constants. Nevertheless any choice of $p$ will result in a valid rejection sampler and clearly the
closer $p$ is to $p^*$, the more efficient the sampler is. 
We have chosen to select $p$ by matching the first moment 
of the geometric to the approximate expected value of the COM-Poisson distribution using \eqref{eq:meanvarapprox}. 
Setting the first moment of the geometric distribution equal to the expected value approximation \eqref{eq:meanvarapprox} implies,
\begin{equation}
\frac{1-p}{p} = \mu + \frac{1}{2\nu} - \frac{1}{2}
\qquad \iff\qquad p = \frac{2\nu}{2\mu\nu + 1 + \nu}.
\label{eq:momentmatching}
\end{equation}

Algorithm \ref{alg:exact_cmp_sampler} presents pseudocode of the COM-Poisson sampler for coding on a computer.

\begin{algorithm}[htbp]
	\linespread{1.2}\selectfont
	\SetAlgoHangIndent{0pt}
	\SetKwInput{Input}{Input}
	\Input{Parameter $\theta = (\mu, \nu)$}
	START \\
	\If{$\nu \geq 1$}{
		Sample $y^\prime \sim \text{Poisson}(\mu)$ using the algorithm from \citet{ahrens1982}.\\
		Calculate $B^{\lbrack \nu \geq 1\rbrack}_{f/g}$ using \eqref{eq:bfg} and set $\alpha = \dfrac{(\mu^{y^\prime}/y^\prime	!)^\nu}{B^{\lbrack \nu \geq 1\rbrack}_{f/g}(\mu^{y^\prime}/y^\prime!)}$. \\
		Generate $u \sim \text{Uniform}(0,1)$.\\
		\eIf{$u \leq \alpha$}{\Return $y^\prime$}{GOTO START}
		
	}
	\If{$\nu < 1$}{
		Compute $p = \dfrac{2\nu}{2\mu\nu + 1 + \nu}$. \\
		Sample $y^\prime \sim \text{geometric}(p)$ by first sampling $u_0 \sim \text{Uniform}(0,1)$ and returning $\left\lfloor\frac{\log(u_0)}{\log(1-p)}\right\rfloor$. \\
		Calculate $B^{\lbrack \nu < 1\rbrack}_{f/g}$ using \eqref{eq:bfg} and set $\alpha = \dfrac{(\mu^{y^\prime}/y^\prime!)^\nu}{B^{\lbrack \nu < 1\rbrack}_{f/g}(1-p)^{y^\prime}p}$. \\
		Generate $u \sim \text{Uniform}(0,1)$.\\
		\eIf{$u \leq \alpha$}{\Return $y^\prime$}{GOTO START}
	}
	\textbf{Note:} For multiple draws from a COM-Poisson$(\mu,\nu)$ distribution, calculate the bound $B_{f/g}^{[\nu \geq1]}$ or $B^{[v<1]}_{f/g}$ once and then use it as an input to the algorithm.
	\caption{Sampler for the COM-Poisson$(\mu,\nu)$ distribution}
	\label{alg:exact_cmp_sampler}
\end{algorithm}

\subsection{Efficiency and comparison of the rejection sampling algorithm}
\label{sec:efficiency_comparison}

To analyse the efficiency of the COM-Poisson sampler we conducted an experiment in which the sampler was run for different values of $\mu$ each over a fine grid of $\nu$ values $\in [0,6]$, as outlined in Figure~\ref{fig:cmpsamplereff}. Here we monitored the acceptance rate of the COM-Poisson rejection sampler, based on the average of the reciprocal of the number of draws from the envelope distribution needed to yield a single draw from the COM-Poisson distribution over $100,000$ independent runs at each grid point. This illustrates that the sampler can achieve acceptance rates between $50\%$ and $100\%$ for $\nu > 1$ (corresponding to a Poisson envelope distribution) and acceptances rates between $30\%$ and $80\%$ for $\nu < 1$ (corresponding to a Geometric envelope distribution). For $\nu< 1$, the acceptance of the sampler increases as $\nu \to 0$. This occurs because the COM-Poisson distribution, with parameter $\nu = 0$, is exactly a geometic distribution with parameter $p = \mu^\nu$ provided $\mu^\nu < 1$. Thus the geometric proposal matches closely the COM-Poisson distribution for small $\nu$.
	\label{sec:cmpsamplereff}
	\begin{figure}[htb]
		\centering
		\includegraphics[width=\textwidth]{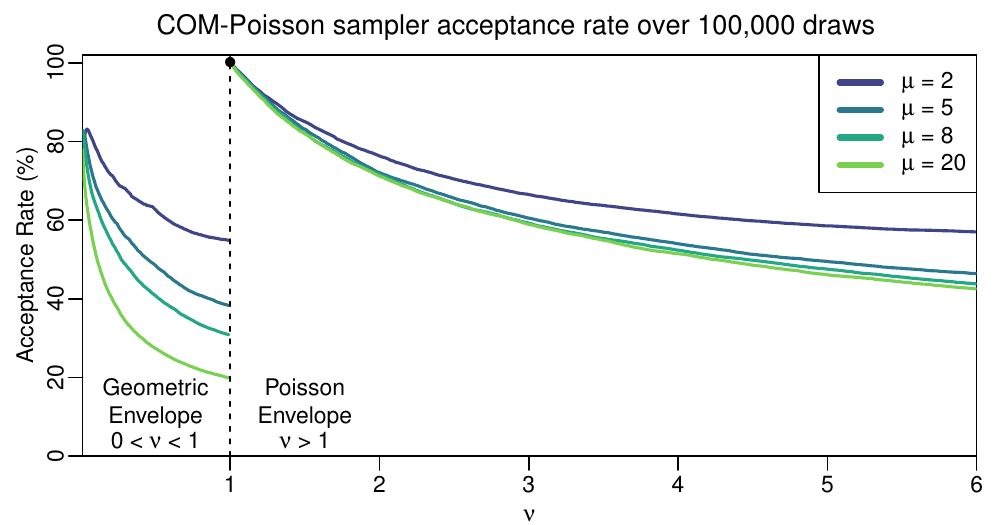}
		\caption{\textbf{COM-Poisson sampler efficiency:} The acceptance rate is the inverse of the number of draws required from either a geometric or a Poisson to obtain one draw from a COM-Poisson distribution for varying values of $\nu$ and $\mu$.}
		\label{fig:cmpsamplereff}
\end{figure}

%
%

In Figure~\ref{fig:speed_up} we compare the efficiency of our COM-Poisson rejection sampler to the rejection sampler of \citet{chanialidis2017}. The sampler developed in \citet{chanialidis2017} involves constructing a tight envelope distribution around the COM-Poisson distribution using a piecewise envelope distribution constructed from four truncated geometric distributions. This piecewise geometric distribution is then sampled using the inversion method for truncated geometric distributions. These truncated geometric distributions are designed to match closely to the COM-Poisson density function, however the computational cost of constructing this envelope is significant and can impact on the CPU time. Our COM-Poisson sampler described in Algorithm \ref{alg:exact_cmp_sampler} is much simpler in construction and requires only a single envelope component depending on the value of $\nu$. Although the single envelope can result in a moderately higher rejection rate compared to \citet{chanialidis2017} and as illustrated in Figure~\ref{fig:cmpsamplereff}, we describe an experiment to illustrate that a potentially higher rejection rate is offset, in CPU time, by the simplicity of constructing and sampling from the single envelope distribution. For example, spending time searching for $p^*$ may improve the acceptance rate; but this search is costly and will reduce the number of effective draws per unit time.

In this experiment we compare the computational run time of the sampler in Algorithm~\ref{alg:exact_cmp_sampler} to that presented in \citet{chanialidis2017} over a wide range of $\mu$ and $\nu$ values. To do this, a $128 \times 128$ grid of equally-spaced $\mu \in [1,25]$ and $\nu \in [0.01, 10.0]$ values was constructed. At each site $2500$ values were drawn from each sampler and the average time taken to draw from each sampler is computed. The plot in Figure~\ref{fig:speed_up} illustrates that Algorithm~\ref{alg:exact_cmp_sampler} is between $1.01$ and $8.01$ times faster. From this one can observe that our sampler represents a significant speedup when used in the context of an MCMC algorithm, as outlined in Section~\ref{sec:comparesampeff}.
This analysis outlines an important aspect of rejection sampling in that the choice of using an efficient envelope must be considered simultaneously with the computational cost of sampling from and constructing this efficient envelope.

\begin{figure}[!ht]
\centering
\includegraphics[width=\textwidth]{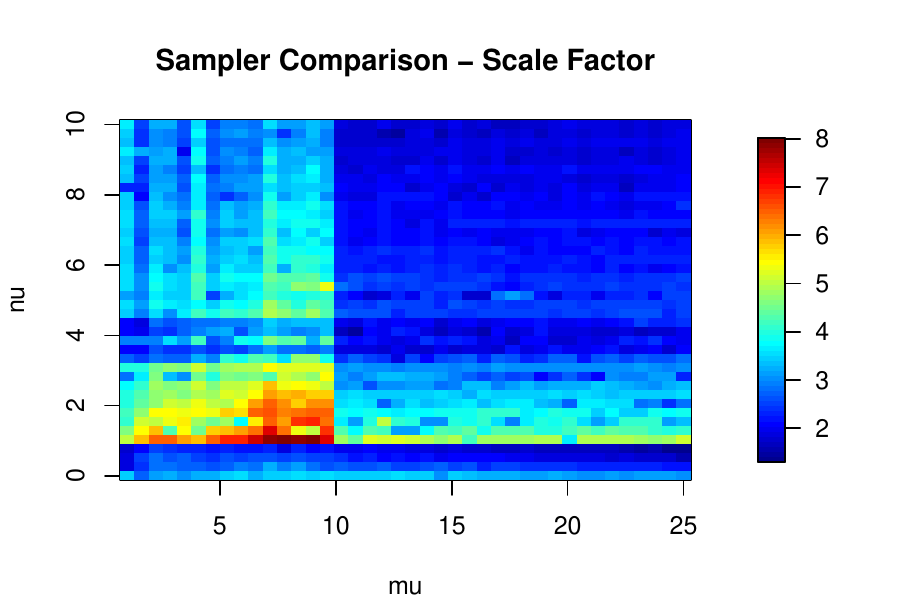}
\caption{This plot shows an experiment to compare the sampler in Algorithm~\ref{alg:exact_cmp_sampler} to that in \citet{chanialidis2017} over a wide range of $\mu$ and $\nu$ values. A $128 \times 128$ grid of $\mu \in [1,25]$ and $\nu \in [0.01, 10.0]$ values is constructed. At each grid point $2500$ values are drawn from each sampler and the average time per sample is computed. The above plot displays the ratio of average sample time from both sampler, where algorithm~\ref{alg:exact_cmp_sampler} is in the dominator of the ratio. At its fastest Algorithm~\ref{alg:exact_cmp_sampler} is 8.01 times faster and at its slowest 1.01 times faster.}
\label{fig:speed_up}
\end{figure}

The speedup of our algorithm in Figure~\ref{fig:speed_up} is notably reduced for $\mu >= 10$. The reason for this reduction is a subtle artefact of computer generation of Poisson samples for large $\mu$. As discussed in \citet{ahrens1982}, computer implementations of Poisson samplers generally switch to a different sampler for values of $\mu \geq 10$. Further work could examine methods to overcome this restriction in order to improve our sampler in this region. Moreover, the results of the experiment in Figure~\ref{fig:speed_up} can be viewed as a lower bound to the potential speedup of our algorithm.

Another benefit of rejection sampling algorithms is that they can be implemented in parallel thus speeding up the computation, however care is needed when implementing parallel random number generators with potential issues discussed in \citet{sprng2000}. Algorithm \ref{alg:exact_cmp_sampler} gives the pseudocode for the COM-Poisson rejection sampler.

\subsection{Unbiased likelihood estimation by monitoring rejection sampling efficiency}
\label{sec:unbiasedlike}
\textbf{Note:} \textit{In this section we reintroduce the index parameter $i$ to discuss rejection sampling at different parameters $\theta_i$, each $\theta_i$ necessarily introduces indexed tractable and intractable sampling bounds and envelope parameters which we denote as $B_{f/g,i}$, $M_{f/g,i}$ and $\gamma_i$ respectively.}

Literature on rejection sampling has explored some methods for recycling the rejected proposals from the envelope distribution. \citet{casella1996} consider reducing the variance of rejection sampling estimates by incorporating the rejected samples in a post-processing step using the Rao-Blackwell theorem. \citet{rao2016} uses the rejected proposals from a rejection sampler in a different way by considering the joint distribution of rejected and accepted draws leading to a method to perform doubly-intractable inference in the case where the likelihood has an associated rejection sampler. Our approach considers yet another use for the rejected proposals by examining the relationship of rejected proposals with the efficiency of the rejection sampler. The efficiency is used to provide an unbiased estimator of the reciprocal normalising constant $\frac{1}{\mathcal{Z}_f(\theta_{i})}$ for each parameter $\theta_i$. An unbiased estimator of the reciprocal normalising constant at each $\theta_i$ will lead directly to an unbiased estimator of the complete likelihood \eqref{eq:fulllike}.

The efficiency of the rejection sampler is critical to the overall performance of the exchange algorithm. The efficiency is defined as the total number of draws $N_1$ required from the envelope distribution until the first acceptance of a draw from $f(y \vert \theta_i)$ or in general the total number of draws $N_r$ required until $r$ acceptances from $f(y \vert \theta_i)$. The number of draws determines how closely the envelope distribution matches the target distribution. The ideal envelope distribution would exactly match the target distribution and reject no proposals, however in this scenario we could simply circumvent rejection sampling and sample from the envelope distribution directly. In Algorithm \ref{alg:vanillarejectionsampling}, run at parameter $\theta_i$, each of the samples proposed from the envelope is followed by a Bernoulli trial with conditional acceptance probability $\alpha_i(y^\ast) = \frac{f(y^{\ast} \vert \theta_{i})}{M_{f/g,i}g(y^{\ast} \vert \gamma_{i})}$. To remove the conditional dependence on $y^\ast$, consider the unconditional acceptance probability $\bar{\alpha}_i$ over all envelope proposals by integrating $\alpha_i(y^\ast)$ with respect to the envelope,
\begin{equation*}
\bar{\alpha}_i = \int_{y^{\ast}} \alpha_i(y^{\ast})  g(y^{\ast} \vert \gamma_{i}) dy^{\ast} = \int_{y^{\ast}} \frac{f(y^{\ast} \vert \theta_{i})}{M_{f/g,i}} \, dy^{\ast} = \frac{1}{M_{f/g,i}}.
\end{equation*} 
It is now obvious that the unconditional acceptance probability and from this the rejection sampler efficiency are controlled directly by the magnitude of the intractable bound $M_{f/g,i}$. On average, one expects to accept a portion $\frac{1}{M_{f/g,i}}$ of proposals from $g(\cdot \vert \gamma_{i})$ as being from $f(\cdot \vert \theta_{i})$ and to discard the remainder. Equivalently, the total number of draws until the first accepted draw follows a geometric distribution with success probability $\frac{1}{M_{f/g,i}}$ and mean $M_{f/g,i}$ i.e.\ $N_1 \sim \text{Geometric}\left(\frac{1}{M_{f/g,i}}\right)$. In general the total number of draws ($N_r$) until $r$ acceptances follows a negative binomial distribution with mean $rM_{f/g,i}$. A natural method to estimate $M_{f/g,i}$ is to run rejection sampling at the parameter $\theta_{i}$ and record the observed total number of draws $n_{r,i}$ required for $r$ acceptances. Then an unbiased estimate of $M_{f/g,i}$ based on $r$ acceptances is given as
\begin{equation*}
\widehat{M}^{(r)}_{f/g,i} = \frac{n_{r,i}}{r}.
\end{equation*}
By the central limit theorem increasing $r$ will lead to lower variance estimates of $M_{f/g,i}$ since $\widehat{M}^{(r)}_{f/g,i}$ can be viewed as the average of $r$ independent geometric trials each having variance $M_{f/g,i}(M_{f/g,i}-1)$. To use $\widehat{M}^{(r)}_{f/g,i}$ to give an unbiased estimate of the likelihood, we have the final assumption that $Z_g(\gamma_{i})$ is known i.e.\ the envelope distribution is tractable.
Then by rewriting \eqref{eq:mfgbfg} a link emerges between the inverse normalising constants of $g( \cdot \vert \gamma_{i})$ and $f( \cdot \vert \theta_{i})$ as
\begin{equation}
\frac{1}{\mathcal{Z}_f(\theta_{i})} = \frac{1}{\mathcal{Z}_g(\gamma_{i})}\frac{M_{f/g,i}}{B_{f/g,i}}. 
\label{eq:invZlink}
\end{equation}
By replacing $M_{f/g,i}$ with its unbiased estimate $\widehat{M}^{(r)}_{f/g,i}$ and multiplying both sides by $q_f(y \vert \theta_{i})$ an unbiased estimator of the likelihood for observation $y_i$ is given as
\begin{equation}
\widehat{f}^{(r)} (y_i \vert \theta_{i}) = \frac{q_f(y_i \vert \theta_{i})}{\mathcal{Z}_g(\gamma_{i})} \frac{\widehat{M}^{(r)}_{f/g,i}}{B_{f/g,i}}
\label{eq:likeest} 
\end{equation}
with all terms on the RHS of \eqref{eq:likeest} known. The estimate of the complete likelihood follows as
\begin{equation}
\widehat{f}^{(r)} (y_{1:n} \vert \theta_{{1:n}}) = \prod_{i=1}^{n} \frac{q_f(y_i \vert \theta_{i})}{\mathcal{Z}_g(\gamma_{i})} \frac{\widehat{M}^{(r)}_{f/g,i}}{B_{f/g,i}}.
\label{eq:fulllikeest} 
\end{equation}
by estimating $\widehat{M}^{(r)}_{f/g,i}$ at each $\theta_i$.
We point out here that the idea to use the acceptance probability of a rejection sampling algorithm to develop an unbiased estimate of the intractable normalising constant has previously explored, in the context of inference for stochastic differential equations. We refer the reader to Section $5$ of \citet{beskos2006}.

This unbiased estimator of the intractable likelihood is very useful especially for model selection as it can be used to construct an estimate of the Bayesian information criterion (BIC) \citep{schwarz1978}.
Bayesian model choice is a challenging problem for models with intractable likelihood functions. The intractability prevents calculation of the likelihood which is required for most model selection criteria.
The BIC estimate from the unbiased likelihood estimator \eqref{eq:fulllikeest} is 
\begin{equation}
\widehat{\text{BIC}}^{(r)} = k\log(n) - 2\log \left(\widehat{f}^{(r)} (y_{1:n} \vert \hat{\theta}_{{1:n}}) \right)
\label{eq:approxBIC}
\end{equation}
where $\hat{\theta}_{1:n}$ maximises the unbiased estimate.
The variance of this BIC estimate depends on $r$ and we find in experiments that $r$ needs to be high ($>1000$) to achieve an accurate estimate of the BIC.
We note that an interesting future research question is to understand the statistical properties, such as consistency of the maximum likelihood estimator $\hat{\theta}_{1:n}$. A useful starting point might be to consider the literature where this issue has been explored for different intractable likelihood problems including 
\cite{geyer1992} and \citet{beskos2006}.

It may also be possible to use the unbiased likelihood estimator (\ref{eq:likeest}) as part of other Bayesian model selection frameworks, for example, 
in the calculation of the marginal likelihood (or evidence) and Bayes factors. This could be a focus of future work. In addition to BIC estimation we 
consider plugging the unbiased estimator directly into the intractable acceptance ratio \eqref{eq:acceptratio} which is the basis of pseudo-marginal 
MCMC algorithms discussed in the next section.

Finally, we note  that \citet{chanialidis2017} presented results comparing models based on the deviance information criterion (DIC) \citep{spieg2002}, without explaining how they overcame the need to evaluate the intractable 
likelihood for each draw from the MCMC sample. Following personal 
communication with the authors it turns out that they approximated the normalising constant, $\mathcal{Z}_f(\mu, \nu)$, using a truncated summation 
involving $k_2-k_1+1$ terms, for positive integers, $k_1$, $k_2$, where $k_1<k_2$, as follows,
\begin{equation}
\mathcal{Z}_f(\mu, \nu) \approx \hat{\mathcal{Z}}_f(\mu, \nu) = \sum_{y=k_1}^{k_2} \left( \frac{\mu^y}{y!} \right)^\nu.
\label{eq:truncsum}
\end{equation}
Here $k_1$ and $k_2$ are selected to satisfy the inequality $k_1 < \lfloor \mu \rfloor < k_2$, where, as before, $\lfloor \mu \rfloor$ is the mode of the COM-Poisson 
distribution. In turn, they approximated the likelihood for each draw from the sample generated from their MCMC algorithm by
\begin{equation}
\hat{f}(y_{1:n} \vert \mu_{1:n}, \nu_{1:n}) = \prod_{i=1}^{n} \left( \frac{\mu_i^{y_i}}{y_i!} \right)^{\nu_i} \frac{1}{\hat{\mathcal{Z}}_f(\mu_i, \nu_i)},
\label{eqn:like_truncated}
\end{equation}
which is then used, for example, to estimate the posterior expected (deviance) log-likelihood which is required for calculation of the DIC. 
Of course, this introduces an approximation of the likelihood into the estimation of the DIC. This raises questions as to how many terms
$k_2-k_1+1$ are needed in order to get an accurate approximation of the likelihood as well as how to choose both $k_1$ and $k_2$ and finally raises the issue 
of the computational cost involved in evaluating the finite truncation for the $n$ terms in the product in (\ref{eqn:like_truncated}). We return to this issue in Section~\ref{sec:comparison_approx_likelihood}, where we find that the number of terms required to get an accurate approximation of the likelihood depends considerably on the values of $\mu$ and $\nu$.


\section{Pseudo-marginal MCMC with the intractable likelihood estimator}
\label{sec:pseudomarginal}
MCMC methods for Bayesian inference require the ability the evaluate the unnormalised posterior distribution. For doubly-intractable problems where the posterior distribution cannot be evaluated as such, there exists pseudo-marginal MCMC algorithms \citep{andrieu2009} which require only an unbiased positive estimate of the posterior distribution. Pseudo-marginal algorithms have been applied to a range of problems such as genetic modelling \citep{beaumont2003} and Markov jump processes \citep{georgoulas2017}. This pseudo-marginal approach can be seen as an alternative to the exchange algorithm which uses exact sampling in place of unbiased posterior estimators. 

We now describe the pseudo-marginal framework and show how our unbiased estimator of the likelihood \eqref{eq:fulllikeest} can be used. Consider a target distribution $\pi(\theta)$ which cannot be pointwise evaluated, but assume that it is possible to construct an unbiased estimate $\hat{\pi}(\theta, z)$ of the target  through the use of auxiliary variables $z \sim w(\cdot \vert \theta)$, $z \in Z$. Using $\hat{\pi}(\theta)$ in place of the true target in a Metropolis-Hastings algorithm leads to the approximate acceptance ratio
\begin{equation*}
\min \left \{1, \frac{\hat{\pi}(\theta^\prime)h(\theta^\prime,\theta)}{\hat{\pi}(\theta) h(\theta, \theta^\prime)} \right \}
\end{equation*}
Pseudo-marginal MCMC algorithms can be viewed as Metropolis-Hastings algorithms over the joint parameter and auxiliary variable state space $\theta \, \cup \, Z$. \cite{andrieu2009} describe two alternative forms: the ``Monte Carlo within Metropolis'' or MCWM algorithm and the ``grouped independence Metropolis-Hastings'' or GIMH algorithm, with only the latter algorithm having the correct target distribution $\pi(\theta)$. In MCWM, both the current $\hat{\pi}(\theta)$ and proposed $\hat{\pi}(\theta^\prime)$ estimate of the target are refreshed every iteration using new auxiliary draws $z$ and $z^\prime$. In GIMH, only the proposed estimate $\hat{\pi}(\theta^\prime)$ is refreshed and $\hat{\pi}(\theta)$ is fixed to the estimate used when $\theta$ was last accepted.

The difficulty with pseudo-marginal algorithms is in constructing an unbiased estimate of the target. Using the rejection sampling unbiased estimator \eqref{eq:fulllikeest} the pseudo-marginal acceptance ratio for the $\beta$ to $\beta^\prime$ move in the COM-Poisson model is
\begin{align}
\alpha_{PM,n}(\beta, \beta^\prime) &= \min \left\{1,\frac{\widehat{f}^{(r)} (y_{1:n} \vert \theta^\prime_{1:n})}{\widehat{f}^{(r)} (y_{1:n} \vert \theta_{{1:n}})}
\frac{h(\beta^\prime, \beta)}{h(\beta, \beta^\prime)} \frac{\pi(\beta^\prime)}{\pi(\beta)} \right\}  \nonumber
\\[2ex] &= \min \left\{1,\frac{\prod_{i=1}^{n} \frac{q_f(y_i \vert \theta^\prime_{i})}{\mathcal{Z}_g(\gamma^\prime_{i})} \frac{\widehat{M}^{\prime,(r)}_{f/g,i}}{B^\prime_{f/g,i}}}{\prod_{i=1}^{n} \frac{q_f(y_i \vert \theta_{i})}{\mathcal{Z}_g(\gamma_{i})} \frac{\widehat{M}^{(r)}_{f/g,i}}{B_{f/g,i}}}\frac{h(\beta^\prime, \beta)}{h(\beta, \beta^\prime)} \frac{\pi(\beta^\prime)}{\pi(\beta)} \right\} \nonumber
\\[2ex]& = \min \left\{1,\frac{\prod_{i=1}^{n} q_f(y_i \vert \theta^\prime_{i})\mathcal{Z}_g(\gamma_{i}) \widehat{M}^{\prime,(r)}_{f/g,i}B_{f/g,i}}{\prod_{i=1}^{n} q_f(y_i \vert \theta_{i})\mathcal{Z}_g(\gamma^\prime_{i}) \widehat{M}^{(r)}_{f/g,i}B^\prime_{f/g,i}}\frac{h(\beta^\prime, \beta)}{h(\beta, \beta^\prime)} \frac{\pi(\beta^\prime)}{\pi(\beta)} \right\}
\label{eq:pseudomarg}
\end{align}
This acceptance ratio is entirely tractable once the estimates $\widehat{M}^{(r)}_{f/g,i}$ have been computed by running rejection sampling at each $\theta^\prime_i$ for the GIMH algorithm or at both $\theta_i$ and $\theta^\prime_i$ for the MCWM algorithm.

\section{Results}
\label{sec:results}
We present two examples of COM-Poisson modelling. The first example is a retail inventory dataset from \citet{shmueli2005} which contains no covariate information. This first example is a toy example that demonstrates how the BIC can be approximated using the unbiased likelihood estimate. The second example is a takeover bids dataset which contains covariates which we model using a COM-Poisson regression model. The second example will firstly show the reduction in computation time available using our sampler and also the ability to choose different regression models using the unbiased likelihood estimate. 

\subsection{Inventory Data}
\label{sec:inventory}
\begin{figure}[!htbp]
	\centering
	\includegraphics[width=\textwidth]{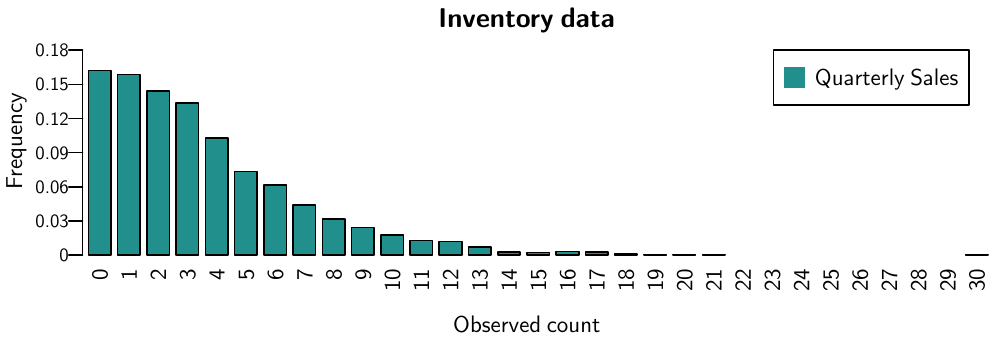}
	\caption{\textbf{Inventory data:} Sales of a particular item of clothing from a well-known clothing brand. The sample mean sales is 3.56 with sample variance 11.31.}
	\label{fig:inventorydata}
\end{figure}
Inventory data or stock count data is of high importance to retailers as it facilitates how they can make decisions regarding future stock levels. The counts we observe 
in this data are the quarterly sales counts of a particular item of clothing in each store across $n=3{,}168$ retail stores. A frequency barplot of the data is shown in Figure~\ref{fig:inventorydata}. The maximum quarterly sales of the clothing item sold was 30 in one store. No stores sold between 22 and 29 items and many stores had 0 sales of the clothing item. The sample mean sales is 3.56, sample variance 11.31 and the sample dispersion index is 3.18, thus overdispersed. This data was analysed by \citet{shmueli2005} using a maximum likelihood approach.

Since this data contains no extra covariate information, we model $\mu$ and $\nu$ directly rather than considering them as functions of covariates. This could be 
viewed as using the link function with only intercept where $\mu = \exp(\beta_{\mu,0}) $ and $\nu = \exp(\beta_{\nu,0}) $ but we consider it easier to model 
$\mu$ and $\nu$. We place Gamma(1, 1) and Gamma(0.0625, 0.25) priors on $\mu$ and $\nu$, respectively. We run the exchange algorithm for 200{,}000 iterations and discard the first 5{,}000 as a burn-in. Each parameter $\mu$ and $\nu$ were updated using a single-site update.
The results are shown in Table~\ref{table:inventory_results} with the target acceptance rate of $44\%$ for single-site updates achieved.
\begin{table}[htbp]
	\centering
	\caption{\textbf{Inventory data posterior estimates:} Posterior estimates from the exchange algorithm. The posterior estimate for $\nu$ shows overdispersion in data.}
	\begin{tabular}{@{\quad}cccc@{\quad}}
		\toprule
		& \textbf{Posterior Mean} & \textbf{Posterior SD} & \textbf{Accept Rate}\\ \midrule
		\multicolumn{1}{c}{$\mu$} & 0.8243     & 0.1444  & 44.39\%  \\ 
		\multicolumn{1}{c}{$\nu$} & 0.1286     & 0.0119    & 41.68\%     \\   \bottomrule	
	\end{tabular}
	\label{table:inventory_results}
\end{table}

\begin{figure}[!htbp]
	\centering
	\includegraphics[width=\textwidth]{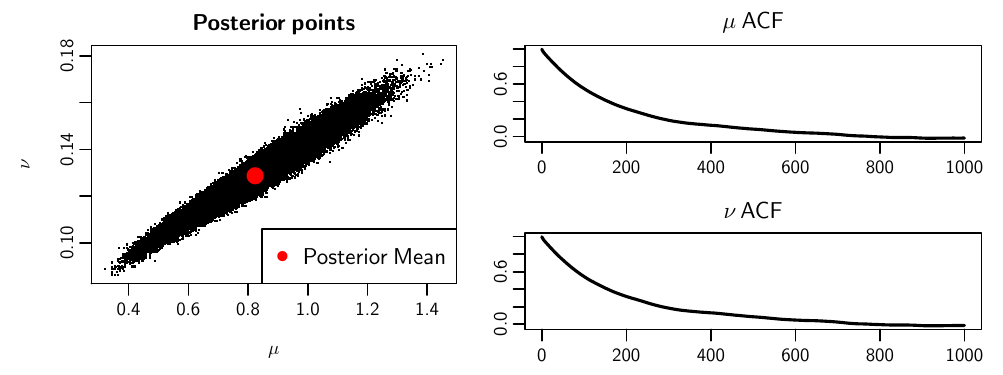}
	\caption{\textbf{Inventory data posterior visualisation:} The distribution of MCMC points for $(\mu,\nu)$ is shown along with posterior mean (left). The estimated autocorrelation function for $\mu$ (top right) and $\nu$ (bottom right).  There is high correlation between $\mu$ and $\nu$ in this example.}
\end{figure}

We now compare the COM-Poisson distribution to the Poisson distribution using $\widehat{\text{BIC}}^{(r)}$ with $r=5000$. The BIC for the Poisson distribution can be calculated exactly since the likelihood is tractable. The comparison of both models is shown in Table~\ref{table:bicinventory} where the COM-Poisson is shown to outperform the Poisson under BIC.

\begin{table}[htbp]
	\centering
	\caption{\textbf{Inventory data:} The COM-Poisson model significantly outperforms the Poisson model using the BIC. The estimate BIC for the COM-Poisson was calculated using $\widehat{\text{BIC}}^{(r)}$ with $r=5000$.}
	\begin{tabular}{@{\quad}lll@{\quad}}
		\toprule
		& \textbf{Poisson} & \textbf{COM-Poisson} \\
		\midrule
		BIC & 17927.68 & \textbf{15067.39}\\ \bottomrule
	\end{tabular}
	\label{table:bicinventory}
\end{table}

\subsubsection{Comparing our method to an approximate (truncated) likelihood approach}
\label{sec:comparison_approx_likelihood}
We now compare our exact approach to instead carrying out MCMC with the true likelihood in the posterior distribution replaced by an approximate likelihood using a truncated normalising constant to $m$ terms
\begin{equation}
\mathcal{Z}_f(\mu, \nu) \approx \hat{\mathcal{Z}}^{m}_f(\mu, \nu) = \sum_{y=0}^{m-1} \left( \frac{\mu^y}{y!} \right)^\nu.
\label{eqn:truncated_nc}
\end{equation}
Such a plug-in approximation of the likelihood leads to one carrying out MCMC on an approximate posterior distribution. In turn, this leads to the important issue of understanding the ergodicity of the resulting Markov chain, an issue which is attracting much interest in the literature, for example, the noisy MCMC approach of \citet{alquier2016} and the Russian roulette approach of \citet{lyne2015}. Therefore, an immediate question arising from using an approximation of the normalising constant, as detailed above~(\ref{eqn:truncated_nc}), would be to understand the convergence properties of the resulting noisy MCMC algorithm although it is beyond the scope of this paper. Of course, the exact sampling algorithm developed in this paper avoids such approximation. 

Instead, we carry out a small study to explore empirically estimated posterior means and standard deviations arising from the noisy MCMC algorithm derived by replacing the true likelihood with the approximate likelihood with truncated normalising constant. To do this, we again analyse the Inventory dataset experimenting with two choices for the number of terms $m$ in the truncated normalising constant~(\ref{eqn:truncated_nc}). 
In particular, we set $m=3{,}300$ as this lead to a similar computation time for one transition of the exact MCMC algorithm. 
Table~\ref{table:inventory_results_trunc} shows that the results of this experiment are in good agreement with the results of the exact MCMC algorithm detailed in Table~\ref{table:inventory_results}. We also note that the truncation of the normalising constant to $m=100$ terms yielded identical results to that of $m=3{,}300$, but resulted in an $86\%$ decrease in the time required by our exact sampler. 
However these results disguise the fact that $m=100$ or even $m=3{,}300$ terms are not sufficient to give an accurate likelihood approximation for all regions of the parameter space. To investigate, we began the algorithm at the starting point $(\mu=500, \nu=0.0001)$ and found that for both cases $m=100$ and $m=3{,}300$ the resulting noisy MCMC algorithms failed to converge whereas our algorithm converged to the posterior estimates as in Table~\ref{table:inventory_results}. The failure to converge is due to an insufficient number of terms in the truncation, meaning the truncation is only practical in some but not all areas of the parameter space. Of course, this has immediate implications when one implements this truncated normalised constant within an MCMC algorithm, as above, since the algorithm may visit areas of the parameter space for which the truncation is inaccurate. 
We also emphasise that one would not know in advance the appropriate value of $m$ to use. This is illustrated in Figure~\ref{fig:truncnormaccuracy} where we explore the accuracy of (\ref{eqn:truncated_nc}) for different values of $\mu$ and $\nu$.

\begin{table}[htbp]
	\centering
	\caption{\textbf{Truncated normalising constant estimates:} Posterior estimates for using a truncated normalising constant in the likelihood with $m=3{,}300$ terms. 
	}
	\begin{tabular}{@{\quad}cccccc@{\quad}}
		\toprule
		& \textbf{Posterior Mean} & \textbf{Posterior SD} & \textbf{Accept Rate} & \textbf{m} & \textbf{Time}\\ \midrule
		\multicolumn{1}{c}{$\mu$} & 0.8265     & 0.1448  & 44.05\% & 3{,}300 & 1.00  \\ 
		\multicolumn{1}{c}{$\nu$} & 0.1282     & 0.0119    & 45.71\% & 3{,}300 & 1.00    \\   \bottomrule	
	\end{tabular}
	\label{table:inventory_results_trunc}
\end{table}

\begin{figure}[htbp]
\includegraphics[width=\textwidth]{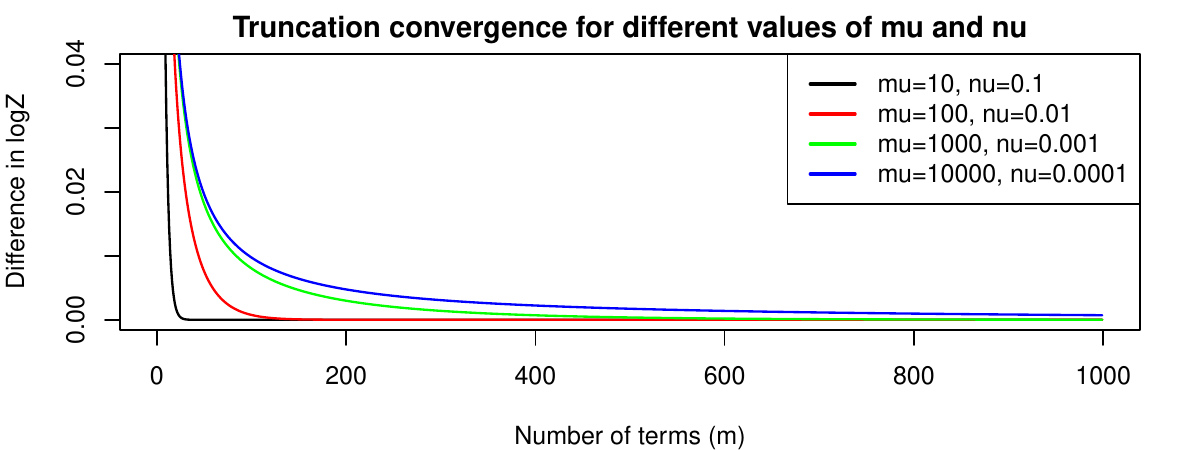}
\caption{\textbf{Truncated normalising constant convergence}: This figure shows that the number of terms required for a truncated normalising constant to converge changes with the value of $\mu$ and $\nu$. The y-axis represents the difference in $\log\left(\hat{\mathcal{Z}}^{m}_f(\mu, \nu)\right)$ for two consecutive values of $m$. This should converge to 0 but for large $\mu$ and small $\nu$ it can take over $1{,}000$ terms to converge.}
\label{fig:truncnormaccuracy}
\end{figure}

\subsection{Takeover bids COM-Poisson regression model}
\label{sec:takeover}

We consider the dataset from \citet{jaggia1993} on the number of bids received by U.S. firms that were targets of bid offers during the period 1978-85 and taken over within 1 year of the initial offer. The response variable is the count of takeover bids excluding the initial bid (\texttt{NUMBIDS}) filed for the particular company. The covariates available include variables which capture the firm-specific characteristics and defensive action taken by the firm in response to their initial takeover bid. A detailed list of available predictors with is given in Appendix \ref{app:takeover}. This data has previously been analysed using Poisson regression by \citet{jaggia1993} and series expansions by \cite{cameron1997}. We reconsider this analysis by comparing two Poisson regression models to three COM-Poisson regression models. The three models considered in our analysis are listed in Table~\ref{table:takeovermodels}. We performed a prior variable selection by fitting a Poisson GLM by maximum likelihood to the data to establish which predictors to consider in the Bayesian GLM analysis. The first model and simplest considers \texttt{NUMBIDS} in a Poisson GLM where the linear predictor of $\log \left(\mu_i\right)$ contains an intercept and two predictors (\texttt{BIDPREM} and \texttt{WHTKNGHT}). The second model adds an extra linear predictor, \texttt{SIZE}, to Model 1. For COM-Poisson regression we consider 3 models, Model 3 moves the covariate \texttt{SIZE} into the $\nu_i$ link function. Model 4 considers dropping \texttt{BIDPREM} from the $\mu_i$ link function and finally Model 5 considers adding another covariate $\texttt{FINREST}$ into the $\nu_i$ link function. For simplicity we use $\beta$ to denote coefficients within the $\mu_i$ link function and $\rho$ to denote coefficients within the $\nu_i$ link function.

Priors for $\beta$ and $\rho$ were set as $\mathcal{N}(0, 5^2)$. The exchange algorithm was run for 100{,}000 observations with a burnin of 10{,}000 draws discarded. Each coefficient was updated using a single-site update with a random walk proposal. The proposal variance of the random walk was tuned to achieve an acceptance rate of 44\% for each site. The posterior results for all models are presented in Table \ref{table:glmresults}. The BIC was calculated using \eqref{eq:approxBIC} with $r$ set to 5000 to give an accurate estimate of the BIC. Model 5 was selected as the best model with the lowest BIC of 386.40. All 3 COM-Poisson regression models are selected as better than the Poisson models, Model 1 and Model 2 suggesting that including covariates in the dispersion link function is worthwhile. 

\begin{table}[htbp]
	\centering
	\caption{\textbf{Takeover bids candidate models:} The $k$ column gives the number of parameters estimated which will be used directly in the calculation of the BIC. Descriptions of all covariates are given in Appendix \ref{app:takeover}.}
	\begin{tabular}{@{\quad}lclr@{\quad}}
		\toprule
		\textbf{Model} & \textbf{Response} ($\bm{y_i}$) & \textbf{Linear Predictor(s)} & \textbf{$\bm{k}$} \\ \midrule
		1 (Poisson) & \texttt{NUMBIDS} &  $\log(\mu_i) = \beta_0 + \beta_1\texttt{BIDPREM} + \beta_2\texttt{WHTKNGHT}$ & 3 \\[2ex]
		2 (Poisson) & \texttt{NUMBIDS} &  $\log(\mu_i) = \beta_0 + \beta_1\texttt{BIDPREM} + \beta_2\texttt{WHTKNGHT} + \beta_3\texttt{SIZE}$ & 4 \\[2ex]
		\multirow{2}{*}{3 (COM-Poisson)} & \multirow{2}{*}{\texttt{NUMBIDS}} &  $\log(\mu_i) = \beta_0 + \beta_1\texttt{BIDPREM} +\beta_2\texttt{WHTKNGHT}$ &  \multirow{2}{*}{5}  \\
		&& $\log(\nu_i) = \rho_0 + \rho_1\texttt{SIZE}$&  \\[2ex]
		\multirow{2}{*}{4 (COM-Poisson)} & \multirow{2}{*}{\texttt{NUMBIDS}} &  $\log(\mu_i) = \beta_0 +\beta_2\texttt{WHTKNGHT}$ &  \multirow{2}{*}{4}  \\ 
		&& $\log(\nu_i) = \rho_0 + \rho_1\texttt{SIZE}$&  \\[2ex]
		\multirow{2}{*}{5 (COM-Poisson)} & \multirow{2}{*}{\texttt{NUMBIDS}} &  $\log(\mu_i) = \beta_0 + \beta_1\texttt{BIDPREM} +\beta_2\texttt{WHTKNGHT}$ &  \multirow{2}{*}{6}  \\ 
		&& $\log(\nu_i) = \rho_0 + \rho_1\texttt{SIZE} + \rho_2\texttt{FINREST}$&  \\\bottomrule       
	\end{tabular}
	\label{table:takeovermodels}
\end{table}

\begin{table}[htbp]
	\centering
	\caption{\textbf{Takeover bids:} Posterior parameter estimates are shown for all 5 models. The lowest BIC is shown in \textbf{bold}. The COM-Poisson regression models provide the best fit to the data according to BIC with all COM-Poisson models having a BIC lower than the Poisson models.}
	\begin{tabular}{@{\quad}cccccc@{\quad}} \toprule
		\multirow{2}{*}{\textbf{Parameter}} & \textbf{Model 1}                      & \textbf{Model 2}                  &\textbf{Model 3} &   \textbf{Model 4}  & \textbf{Model 5}                 \\
		&\textbf{Mean (SD)} & \textbf{Mean (SD)} & \textbf{Mean (SD)}  & \textbf{Mean (SD)} &  \textbf{Mean (SD)}\\  \midrule
		
		\multicolumn{1}{c}{$\hat{\beta_0}$} & \phantom{-}1.130 (0.505) & \phantom{-}1.063 (0.532) & \phantom{-}1.077 (0.384) & \phantom{-}0.329 (0.100) &  \phantom{-}0.354 (0.091) \\[0.5ex]
		
		\multicolumn{1}{c}{$\hat{\beta_1}$} & -0.728 (0.368) & -0.713 (0.382) & -0.553 (0.281) & - & - \\[0.5ex]
		
		\multicolumn{1}{c}{$\hat{\beta_2}$} & \phantom{-}0.583 (0.152) & \phantom{-}0.576 (0.152) & \phantom{-}0.458 (0.110) & \phantom{-}0.463 (0.111)  & \phantom{-}0.431 (0.103) \\[0.5ex]
		
		\multicolumn{1}{c}{$\hat{\beta_3}$} & - & \phantom{-}0.035 (0.017) & - & - & - \\[0.5ex]
		
		\multicolumn{1}{c}{$\hat{\rho_0}$} & - & -  & \phantom{-}0.674 (0.175) & \phantom{-}0.646 (0.175) & \phantom{-}0.789 (0.179)    \\[0.5ex]
		
		\multicolumn{1}{c}{$\hat{\rho_1}$} & - & -  & -0.171 (0.051) & -0.174 (0.052) & -0.176 (0.049)   \\[0.5ex] 
		
		\multicolumn{1}{c}{$\hat{\rho_2}$} & - & -  & - & - & -0.952 (0.448)  \\[0.5ex] \midrule
		
		\multicolumn{1}{c}{$\widehat{\textbf{BIC }} (r = 5000)$} & $397.49$ & $398.32$ & $386.89$ & $386.98$ & $\mathbf{386.40}$ \\[0.5ex]
		
		\multicolumn{1}{c}{\textbf{Rank}} & 4 & 5 (Worst) & 2 & 3 & \textbf{1 (Best)}   \\ \bottomrule                                      
	\end{tabular}
	\label{table:glmresults}
\end{table}

\subsubsection{Comparing sampler efficiency}
\label{sec:comparesampeff}
Following Section~\ref{sec:efficiency_comparison}, here we compare the efficiency of Algorithm~\ref{alg:exact_cmp_sampler} to the rejection sampler of \citet{chanialidis2017} in the context of the takeover bids dataset. To do so, we examine the number of MCMC draws per second resulting from Algorithm~\ref{alg:exact_cmp_sampler} and from the rejection sampler of \citet{chanialidis2017} for each COM-Poisson models (Models 3, 4 and 5) outlined in Table \ref{table:takeovermodels}. Our sampler provides a threefold increase in the number of MCMC draws per second in each of the COM-Poisson models agreeing with the experiment presented in Figure~\ref{fig:speed_up}. 

\begin{table}[htbp]
\centering
\caption{\textbf{{Algorithm efficiency comparison:}} This table shows the number of MCMC draws per second for our rejection sampler (Algorithm \ref{alg:exact_cmp_sampler}) compared with the rejection sampler of \citet{chanialidis2017}. Our algorithm is at least 3 times faster for all 3 models.}
\begin{tabular}{@{\quad}lccc@{\quad}} \toprule
 & \textbf{Model 3} & \textbf{Model 4} & \textbf{Model 5}  \\ \midrule
MCMC draws per sec. (Algorithm \ref{alg:exact_cmp_sampler}) & 9{,}973 & 12{,}534 & 9{,}297\\
MCMC draws per sec. (\citet{chanialidis2017}) & 2{,}866 & 3{,}585 & 2{,}828 \\ \midrule
\textbf{Relative efficiency} & 3.480 & 3.496 & 3.287 \\
 \bottomrule
 \label{table:samplercompare}
\end{tabular}
\end{table}

\subsection{Pseudo-marginal MCMC results}
We choose Model 5 in Table \ref{table:takeovermodels} to implement the GIMH pseudo-marginal MCMC algorithm using the unbiased likelihood estimator \eqref{eq:fulllikeest} in the acceptance ratio. The unbiased likelihood estimator was calculated for 5 scenarios with $r=1,5, 10, 50,100$. The posterior density estimates for each of the pseudo-marginal MCMC runs are shown in Table \ref{table:pmresults} along with the CPU time and multivariate effective sample size (mESS) \citep{vats2015}. The mESS is a generalisation of the effective sample size (ESS), which captures the cross-autocorrelation between parameters in the MCMC output.
\begin{equation*}
\text{mESS} = n_\text{mcmc} \left(\frac{\lvert \Lambda \rvert}{\lvert \Sigma \rvert}\right)^{1/p}
\end{equation*}
where $n_\text{mcmc}$ is the number of MCMC draws, $\Lambda$ is the sample autocovariance matrix of the MCMC output and $\Sigma$ is the true covariance of the posterior 
which can be estimated by a batch-means method described in \citet{vats2015}. Figure \ref{fig:pseudomcmc} compares the pseudo-marginal ($r=100$) posterior MCMC chain, 
density and autocorrelation with the exchange algorithm.
There is close agreement between the two methods, however the CPU time for the pseudo-marginal algorithm is dramatically increased.

\begin{table}[!htbp]
\centering
\caption{\textbf{Pseudo-marginal MCMC results:} The results for the pseudo-marginal MCMC algorithm run with the unbiased likelihood estimator.}
\begin{tabular}{@{\quad}c@{\,}cccccc@{\quad}} \toprule
\multirow{3}{*}{} & \multirow{1}{*}{\textbf{Exchange}}  & \multicolumn{5}{c}{\textbf{Pseudo-marginal MCMC - GIMH}} \\
 &  & {$r=1$} & {$r=5$}  & $r=10$ & $r=50$  & $r=100  $         \\

 & \textbf{Mean ($\sigma$)} & \textbf{Mean ($\sigma$)} & \textbf{Mean ($\sigma$)}  & \textbf{Mean ($\sigma$)} & \textbf{Mean ($\sigma$)}  & \textbf{Mean ($\sigma$)}\\  \midrule

\multicolumn{1}{c}{$\hat{\beta_0}$} & \phantom{-}0.354 (0.09)   & \phantom{-}0.337 (0.08) & \phantom{-}0.361 (0.09) & \phantom{-}0.354 (0.09) & \phantom{-}0.354 (0.09) & \phantom{-}0.356 (0.09) \\[0.5ex]


\multicolumn{1}{c}{$\hat{\beta_2}$} & \phantom{-}0.431 (0.10)   & \phantom{-}0.446 (0.10) & \phantom{-}0.424 (0.10) & \phantom{-}0.432 (0.10) & \phantom{-}0.433 (0.10) & \phantom{-}0.431 (0.11) \\[0.5ex]


\multicolumn{1}{c}{$\hat{\rho_0}$} & \phantom{-}0.789 (0.18)   & \phantom{-}0.785 (0.14) & \phantom{-}0.796 (0.17) & \phantom{-}0.790 (0.17) & \phantom{-}0.794 (0.18) & \phantom{-}0.793 (0.18) \\[0.5ex]

\multicolumn{1}{c}{$\hat{\rho_1}$} & -0.176 (0.05)   & -0.176 (0.05) & -0.172 (0.05) & -0.178 (0.05) & -0.175 (0.05) & -0.175 (0.05) \\[0.5ex]
\multicolumn{1}{c}{$\hat{\rho_2}$} & -0.952 (0.45)   & -0.817 (0.50) & -0.981 (0.44) & -0.944 (0.45) & -0.962 (0.45) & -0.955 (0.45)\\[0.5ex]

\midrule
\multicolumn{1}{c}{\textbf{CPU(s)}} & 10.617 & 46.257 & 188.050 & 355.001  & 1668.14 & 3512.57 \\
\multicolumn{1}{c}{\textbf{mESS}} & 4{,}962 & 345 & 904 & 2701  & 8{,}820 & 10{,}922 \\
     \bottomrule                                      
\end{tabular}
\label{table:pmresults}
\end{table}

\begin{figure}[htbp]
	\centering
	\includegraphics[width=\textwidth]{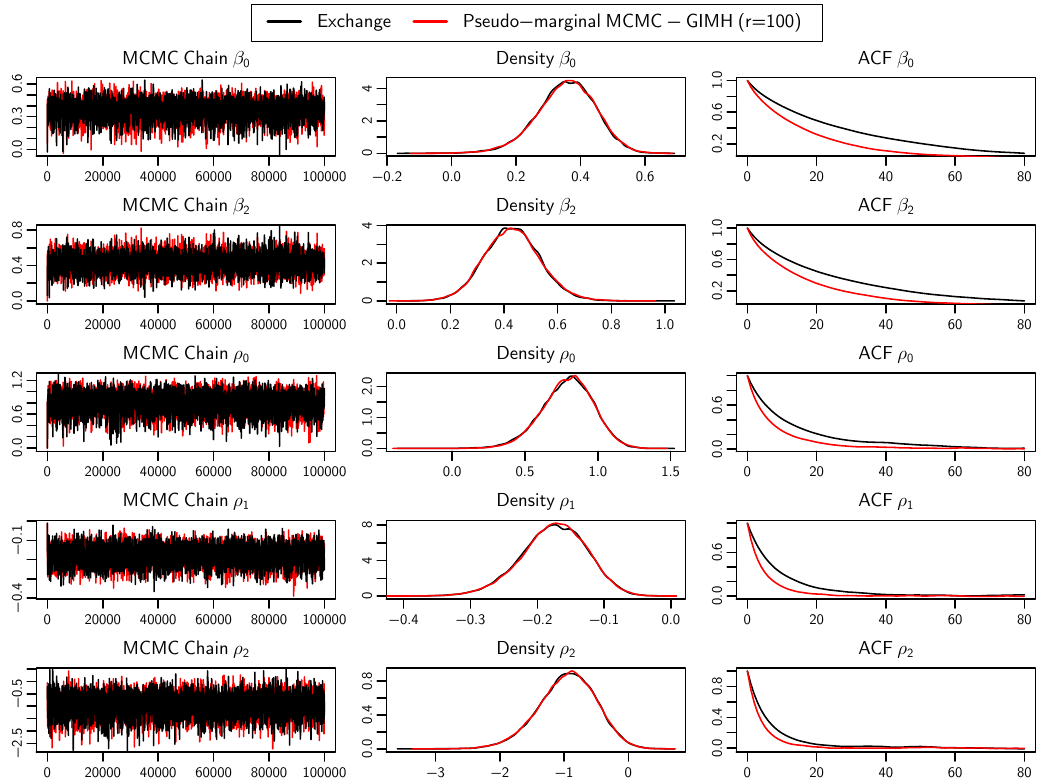}
	\caption{\textbf{Comparison of the pseudo-marginal MCMC algorithm and the exchange algorithm:} Both algorithms show agreement in posterior density estimates Pseudo-marginal MCMC shows a reduction in autocorrelation compared to the exchange algorithm for all parameters however this is at the expense of increased computational time shown in Table \ref{table:pmresults}.}
	\label{fig:pseudomcmc}
\end{figure}

\section{Discussion}
\label{sec:discussion}
This paper provides a new rejection sampler to sample from the COM-Poisson distribution.
This rejection sampler allows for faster parameter inference to be performed in COM-Poisson regression models compared with the sampler of \citet{chanialidis2017}.
Our rejection sampling algorithm shows that one must consider both the rejection rate of the enveloping distribution and the sampling time from the enveloping distribution in order to construct an efficient rejection sampler.
We have also shown how the number of rejected proposals within rejection sampling can be used to construct an unbiased intractable likelihood estimator. This estimator can be used to perform model selection using BIC but future work could include this estimator within other Bayesian model choice strategies. 
Model selection for doubly-intractable problems such as COM-Poisson regression is a difficult problem and our approach offers one solution to perform model selection. 
The unbiased likelihood estimator was shown to work well when used within a pseudo-marginal MCMC algorithm \citep{andrieu2009} as it provided an unbiased estimate of 
the acceptance ratio. This unbiased likelihood estimator could be constructed for other intractable models where a rejection sampling algorithm exists to sample from
the likelihood, presenting opportunities for future research.

\section*{Acknowledgements}
The Insight Centre for Data Analytics is supported by Science Foundation Ireland under Grant Number
12/RC/2289$\_$P2.

\clearpage
\bibliographystyle{chicago}
\bibliography{cmp_modelselection_arxiv.bib}

\begin{thebibliography}{}

\bibitem[\protect\citeauthoryear{Ahrens and Dieter}{Ahrens and
  Dieter}{1982}]{ahrens1982}
Ahrens, J.~H. and U.~Dieter (1982).
\newblock Computer generation of {Poisson} deviates from modified {Normal}
  distributions.
\newblock {\em ACM Transactions on Mathematical Software\/}~{\em 8\/}(2),
  163--179.

\bibitem[\protect\citeauthoryear{Alquier, Friel, Everitt, and Boland}{Alquier
  et~al.}{2016}]{alquier2016}
Alquier, P., N.~Friel, R.~Everitt, and A.~Boland (2016).
\newblock Noisy {Monte Carlo}: convergence of {Markov} chains with approximate
  transition kernels.
\newblock {\em Statistics and Computing\/}~{\em 26\/}(1), 29--47.

\bibitem[\protect\citeauthoryear{Andrieu and Roberts}{Andrieu and
  Roberts}{2009}]{andrieu2009}
Andrieu, C. and G.~O. Roberts (2009).
\newblock The pseudo-marginal approach for efficient {Monte} {Carlo}
  computations.
\newblock {\em The Annals of Statistics\/}~{\em 37\/}(2), 697--725.

\bibitem[\protect\citeauthoryear{Beaumont}{Beaumont}{2003}]{beaumont2003}
Beaumont, M.~A. (2003).
\newblock Estimation of population growth or decline in genetically monitored
  populations.
\newblock {\em Genetics\/}~{\em 164\/}(3), 1139--1160.

\bibitem[\protect\citeauthoryear{Beskos, Papaspiliopoulos, Roberts, and
  Fearnhead}{Beskos et~al.}{2006}]{beskos2006}
Beskos, A., O.~Papaspiliopoulos, G.~O. Roberts, and P.~Fearnhead (2006).
\newblock Exact and computationally efficient likelihood‐based estimation for
  discretely observed diffusion processes (with discussion).
\newblock {\em Journal of the Royal Statistical Society, Series B\/}~{\em
  68\/}(3), 333--382.

\bibitem[\protect\citeauthoryear{Caimo and Friel}{Caimo and
  Friel}{2011}]{caimo2011}
Caimo, A. and N.~Friel (2011).
\newblock Bayesian inference for exponential random graph models.
\newblock {\em Social Networks\/}~{\em 33\/}(1), 41--55.

\bibitem[\protect\citeauthoryear{Cameron and Johansson}{Cameron and
  Johansson}{1997}]{cameron1997}
Cameron, A.~C. and P.~Johansson (1997).
\newblock Count data regression using series expansions: With applications.
\newblock {\em Journal of Applied Econometrics\/}~{\em 12\/}(3), 203--223.

\bibitem[\protect\citeauthoryear{Casella and Robert}{Casella and
  Robert}{1996}]{casella1996}
Casella, G. and C.~P. Robert (1996).
\newblock Rao-blackwellisation of sampling schemes.
\newblock {\em Biometrika\/}~{\em 83\/}(1), 81--94.

\bibitem[\protect\citeauthoryear{Chanialidis, Evers, Neocleous, and
  Nobile}{Chanialidis et~al.}{2018}]{chanialidis2017}
Chanialidis, C., L.~Evers, T.~Neocleous, and A.~Nobile (2018).
\newblock Efficient {Bayesian} inference for {COM-Poisson} regression models.
\newblock {\em Statistics and Computing\/}~{\em 28}, 595--608.

\bibitem[\protect\citeauthoryear{Conway and Maxwell}{Conway and
  Maxwell}{1962}]{conway1962}
Conway, R.~W. and W.~L. Maxwell (1962).
\newblock A queuing model with state dependent service rates.
\newblock {\em Journal of Industrial Engineering\/}~{\em 12\/}(2), 132--136.

\bibitem[\protect\citeauthoryear{Georgoulas, Hillston, and
  Sanguinetti}{Georgoulas et~al.}{2017}]{georgoulas2017}
Georgoulas, A., J.~Hillston, and G.~Sanguinetti (2017).
\newblock Unbiased {Bayesian} inference for population {Markov} jump processes
  via random truncations.
\newblock {\em Statistics and Computing\/}~{\em 27\/}(4), 991--1002.

\bibitem[\protect\citeauthoryear{Geyer and Thompson}{Geyer and
  Thompson}{1992}]{geyer1992}
Geyer, C. and E.~A. Thompson (1992).
\newblock Constrained {M}onte {C}arlo maximum likelihood for dependent data.
\newblock {\em Journal of the Royal Statistical Society, Series B\/}~{\em
  54\/}(3), 657 -- 699.

\bibitem[\protect\citeauthoryear{Gilks and Wild}{Gilks and
  Wild}{1992}]{gilks1992}
Gilks, W.~R. and P.~Wild (1992).
\newblock Adaptive rejection sampling for {Gibbs} sampling.
\newblock {\em Journal of the Royal Statistical Society: Series C (Applied
  Statistics)\/}~{\em 41\/}(2), 337--348.

\bibitem[\protect\citeauthoryear{Guikema and Goffelt}{Guikema and
  Goffelt}{2008}]{guikema2008}
Guikema, S.~D. and J.~P. Goffelt (2008).
\newblock A flexible count data regression model for risk analysis.
\newblock {\em Risk Analysis\/}~{\em 28\/}(1), 213--223.

\bibitem[\protect\citeauthoryear{Hilbe}{Hilbe}{2011}]{hilbe2011}
Hilbe, J.~M. (2011).
\newblock {\em Negative binomial regression}.
\newblock Cambridge University Press.

\bibitem[\protect\citeauthoryear{Jaggia and Thosar}{Jaggia and
  Thosar}{1993}]{jaggia1993}
Jaggia, S. and S.~Thosar (1993).
\newblock Multiple bids as a consequence of target management resistance: A
  count data approach.
\newblock {\em Review of Quantitative Finance and Accounting\/}~{\em 3\/}(4),
  447--457.

\bibitem[\protect\citeauthoryear{Lyne, Girolami, Atchadé, Strathmann, and
  Simpson}{Lyne et~al.}{2015}]{lyne2015}
Lyne, A.-M., M.~Girolami, Y.~Atchadé, H.~Strathmann, and D.~Simpson (2015,
  11).
\newblock On {Russian} roulette estimates for {Bayesian} inference with
  doubly-intractable likelihoods.
\newblock {\em Statistical Science\/}~{\em 30\/}(4), 443--467.

\bibitem[\protect\citeauthoryear{Mascagni and Srinivasan}{Mascagni and
  Srinivasan}{2000}]{sprng2000}
Mascagni, M. and A.~Srinivasan (2000).
\newblock Algorithm 806: {SPRNG}: A scalable library for pseudorandom number
  generation.
\newblock {\em ACM Transactions on Mathematical Software\/}~{\em 26\/}(3),
  436--461.

\bibitem[\protect\citeauthoryear{M{\o}ller, Pettitt, Reeves, and
  Berthelsen}{M{\o}ller et~al.}{2006}]{moeller2006}
M{\o}ller, J., A.~N. Pettitt, R.~Reeves, and K.~K. Berthelsen (2006).
\newblock {An efficient Markov chain Monte Carlo method for distributions with
  intractable normalising constants}.
\newblock {\em Biometrika\/}~{\em 93\/}(2), 451--458.

\bibitem[\protect\citeauthoryear{Murray, Ghahramani, and MacKay}{Murray
  et~al.}{2006}]{murray2006}
Murray, I., Z.~Ghahramani, and D.~J.~C. MacKay (2006).
\newblock {MCMC} for doubly-intractable distributions.
\newblock In {\em Proceedings of the 22nd Annual Conference on Uncertainty in
  Artificial Intelligence (UAI-06)}, pp.\  359--366. AUAI Press.

\bibitem[\protect\citeauthoryear{Propp and Wilson}{Propp and
  Wilson}{1996}]{propp1996}
Propp, J.~G. and D.~B. Wilson (1996).
\newblock Exact sampling with coupled {M}arkov chains and applications to
  statistical mechanics.
\newblock {\em Random Structures and Algorithms\/}~{\em 9\/}(1\&2), 223--252.

\bibitem[\protect\citeauthoryear{Rao, Lin, and Dunson}{Rao
  et~al.}{2016}]{rao2016}
Rao, V., L.~Lin, and D.~B. Dunson (2016).
\newblock Data augmentation for models based on rejection sampling.
\newblock {\em Biometrika\/}~{\em 103\/}(2), 319--335.

\bibitem[\protect\citeauthoryear{Schwarz}{Schwarz}{1978}]{schwarz1978}
Schwarz, G. (1978).
\newblock Estimating the dimension of a model.
\newblock {\em The Annals of Statistics\/}~{\em 6\/}(2), 461--464.

\bibitem[\protect\citeauthoryear{Sellers and Shmueli}{Sellers and
  Shmueli}{2010}]{sellers2010}
Sellers, K.~F. and G.~Shmueli (2010).
\newblock A flexible regression model for count data.
\newblock {\em The Annals of Applied Statistics\/}~{\em 4\/}(2), 943--961.

\bibitem[\protect\citeauthoryear{Shmueli, Minka, Kadane, Borle, and
  Boatwright}{Shmueli et~al.}{2005}]{shmueli2005}
Shmueli, G., T.~P. Minka, J.~B. Kadane, S.~Borle, and P.~Boatwright (2005).
\newblock A useful distribution for fitting discrete data: Revival of the
  {Conway-Maxwell-Poisson} distribution.
\newblock {\em Journal of the Royal Statistical Society: Series C (Applied
  Statistics)\/}~{\em 54\/}(1), 127--142.

\bibitem[\protect\citeauthoryear{Shmueli, Russo, and Jank}{Shmueli
  et~al.}{2007}]{shmueli2007}
Shmueli, G., R.~P. Russo, and W.~Jank (2007).
\newblock The {BARISTA}: A model for bid arrivals in online auctions.
\newblock {\em The Annals of Applied Statistics\/}~{\em 1\/}(2), 412--441.

\bibitem[\protect\citeauthoryear{Spiegelhalter, Best, Carlin, and Van
  Der~Linde}{Spiegelhalter et~al.}{2002}]{spieg2002}
Spiegelhalter, D.~J., N.~G. Best, B.~P. Carlin, and A.~Van Der~Linde (2002).
\newblock {Bayesian} measures of model complexity and fit.
\newblock {\em Journal of the Royal Statistical Society: Series B (Statistical
  Methodology)\/}~{\em 64\/}(4), 583--639.

\bibitem[\protect\citeauthoryear{{Vats}, {Flegal}, and {Jones}}{{Vats}
  et~al.}{2019}]{vats2015}
{Vats}, D., J.~M. {Flegal}, and G.~L. {Jones} (2019).
\newblock {Multivariate Output Analysis for Markov chain Monte Carlo}.
\newblock {\em Biometrika\/}~{\em 106\/}(2), 321--337.

\bibitem[\protect\citeauthoryear{von Neumann}{von
  Neumann}{1951}]{vonneumann1951}
von Neumann, J. (1951).
\newblock {Various Techniques Used in Connection with Random Digits}.
\newblock {\em National Bureau of Standards Applied Mathematics Series\/}~{\em
  12}, 36--38.

\bibitem[\protect\citeauthoryear{Wei and Murray}{Wei and
  Murray}{2017}]{wei2016}
Wei, C. and I.~Murray (2017, 20--22 Apr).
\newblock {Markov Chain Truncation for Doubly-Intractable Inference}.
\newblock In A.~Singh and J.~Zhu (Eds.), {\em Proceedings of the 20th
  International Conference on Artificial Intelligence and Statistics},
  Volume~54 of {\em Proceedings of Machine Learning Research}, Fort Lauderdale,
  FL, USA, pp.\  776--784. PMLR.

\end{thebibliography}
%

\clearpage
\section*{Appendices}
\begin{appendices}
	
	\section{COM-Poisson rejection sampler Proof}
	\label{sec:cmpproof}
	\begin{proof}
		
		There are two cases to consider depending on the value of $\nu$. \\\textit{(In this proof we use the Iverson bracket $[\cdot]$, where $[x] = 1$ if the enclosed statement $x$ is true and 0 if the statement $x$ is false.)} 	
		
		\noindent
		\textbf{Case 1}: When $[\nu \geq 1] = 1$
		
		\noindent
		The enveloping bound as illustrated in Figure~\ref{fig:envelopes} (left) is the ratio of the COM-Poisson mass function to a Poisson mass function
		\begin{equation*}
		B^{[\nu \geq 1]}_{f/g} = \sup_y \left \{ \dfrac{\left(\frac{\mu^y}{y!}\right)^\nu}{\frac{\mu^y}{y!}} \right \} = \sup_y \left \{  \left(\dfrac{\mu^y}{y!} \right)^{\nu-1} \right \}.
		\end{equation*}
		To find the supremum, assume that the supremum occurs at the point $y_m$ and by the unimodality of the COM-Poisson distribution, it will satisfy
		\begin{eqnarray}
		\left(\dfrac{\mu^{y_m}}{y_m!} \right)^{\nu-1} &\geq& \left(\dfrac{\mu^{y_m + h}}{(y_m + h)!} \right)^{\nu-1} \nonumber \\
		\iff \dfrac{\mu^{y_m}}{y_m!} 	&\geq& \dfrac{\mu^{y_m + h}}{(y_m + h)!},
		\label{eqn:inequalities1}
		\end{eqnarray}
		for $h \in \{-y_m, -y_m+1, \dots  \}$.
		To find the value of $y_m$ that will satisfy this collection of inequalities over $h$ it is enough to consider the dominant inequalities at $h = \pm 1$,
		\begin{align}
		\mu-1\leq y_m \leq \mu \qquad \Longrightarrow \, y_m = \lfloor \mu \rfloor.
		\label{eq:inequalitypois}
		\end{align}
		Therefore if $y_m$ satisfies \eqref{eq:inequalitypois}, it will also satisfy (\ref{eqn:inequalities1}) for all $h \in \{-y_m, -y_m+1, \dots  \}$.
		Thus the bound becomes
		\begin{align*}
		B^{[\nu \geq 1]}_{f/g} = \left(\frac{\mu^{\lfloor \mu \rfloor}}{\lfloor \mu \rfloor!}\right)^{\nu-1},
		\end{align*}
		which is a tractable bound and for the special case of Poisson ($\nu = 1$), then $B^{[\nu = 1]}_{f/g}=1$. In the case of integer $\mu$ this bound will also ensure coverage of the dual mode at $\mu$ and $\mu-1$ (see Figure~\ref{fig:envelopes}).
		\bigskip
		
		\noindent\textbf{Case 2}: When $[\nu < 1] = 1$
		
		\noindent
		The enveloping bound as illustrated in Figure~\ref{fig:envelopes} (right) is the ratio of the COM-Poisson mass function to a geometric mass function
		\begin{equation*}
		B^{[\nu < 1]}_{f/g} = \sup_y \left \{ \dfrac{\left(\frac{\mu^y}{y!}\right)^\nu}{(1-p)_{}	^{y}p} \right \} = \frac{1}{p}\sup_y \left \{ \dfrac{\left({\mu^\nu}\right)^y}{(1-p)_{}	^{y} y!^\nu} \right \}.
		\end{equation*}
		For now we assume $p$ is unknown and we prove the general bound for any $p$. A discussion of which value of $p$ to choose in practice	 will be given after this proof. As in \textbf{Case 1}, assume that 
		the supremum occurs at the point $y_m$, which by the unimodality of the COM-Poisson distribution, will satisfy
		\begin{eqnarray}
		\frac{1}{p} \dfrac{\left({\mu^\nu}\right)^{y_m}}{(1-p)_{}	^{y_m} y_m!^\nu} &\geq& \frac{1}{p} \dfrac{\left({\mu^\nu}\right)^{y_m+h}}{(1-p)_{}	^{y_m+h} (y_m+h	)!^\nu} \nonumber \\[2ex]
		\iff  \dfrac{1}{ y_m!} &\geq&  \dfrac{{\mu}^{h}}{(1-p)_{}	^{h/v} (y_m+h	)!},
		\label{eqn:inequalities2}
		\end{eqnarray}
		for $h \in \{-y_m, -y_m+1, \dots  \}$.
		To find the value of $y_m$ that will satisfy this collection of inequalities over $h$ it is enough to consider again only the inequalities at $h = \pm 1$. This leads to
		\begin{equation}
		\frac{\mu}{(1-p)^{{1}/{\nu}}} - 1 \leq y_m \leq \frac{\mu}{(1-p)^{{1}/{\nu}}} \qquad \Longrightarrow \, y_m = \left\lfloor \frac{\mu}{(1-p)^{{1}/{\nu}}} \right\rfloor.
		\label{eq:inequalitygeom}
		\end{equation}
		As for \textbf{Case 1}, if $y_m$ satisfies \eqref{eq:inequalitygeom} it will also satisfy (\ref{eqn:inequalities2}) for all $h \in \{-y_m, -y_m+1, \dots  \}$.
		Thus the bound becomes
		\begin{align*}
		B^{[\nu < 1]}_{f/g} = \frac{1}{p} \dfrac{{\mu}^{\left(\nu\left\lfloor \frac{\mu}{(1-p)^{{1}/{\nu}}} \right\rfloor\right)}}{(1-p)_{}	^{	\left(\left\lfloor \frac{\mu}{(1-p)^{{1}/{\nu}}} \right\rfloor\right)} \left(\left\lfloor \frac{\mu}{(1-p)^{{1}/{\nu}}} \right\rfloor!\right)	^\nu}.
		\end{align*}
	\end{proof}
	
	\section{Takeover bids data}
	\label{app:takeover}
	The variables available for the takeover bids dataset are show in Table~\ref{table:takeover} along with a brief description and possible values. This data was first analysed in \citet{jaggia1993}.
	\begin{table}[htb]
		\centering
					\caption{\textbf{Takeover bids:} Description of variables and possible values.}
		\begin{tabular}{@{}llll@{}}
			\toprule
			\textbf{Name} & \textbf{Description} &\textbf{Values} &\textbf{Type}  \\ \midrule
			\texttt{NUMBIDS}& Number of takeover bids received after initial bid &\texttt{0,1,2,\dots} &Response \\			
			\texttt{LEGLREST}&    Did management try legal defense by lawsuit? &\texttt{1=YES,0=NO} &Predictor \\
			\texttt{REALREST}&    Did management propose changes in asset structure? &\texttt{1=YES,0=NO} &Predictor \\
			\texttt{FINREST}&    Did management propose changes in ownership structure? &\texttt{1=YES,0=NO} &Predictor \\
			\texttt{WHTKNGHT}&    Did management invite friendly third-party bid? &\texttt{1=YES,0=NO} &Predictor \\
			\texttt{REGULATN}&    Did federal regulators intervene? &\texttt{1=YES,0=NO} &Predictor \\	
			\texttt{BIDPREM}&    Bid price divided by price 14 working days before bid. &\texttt{> 0} &Predictor \\
			\texttt{INSTHOLD}& Percentage of stock held by institutions. &\texttt{[0,1]} &Predictor \\
			\texttt{SIZE}& Total book value of assets &\texttt{USD billions} &Predictor \\	\bottomrule					
		\end{tabular}

			\label{table:takeover}
		\end{table}

\end{appendices}

\end{document}